\documentclass[12pt,a4paper]{article}
\usepackage{mathptmx}
\usepackage[utf8]{inputenc}
\usepackage[T1]{fontenc}
\usepackage[margin=0.7in]{geometry}
\usepackage{amsmath}
\usepackage{amsfonts}
\usepackage{amssymb}
\usepackage{graphicx}
\usepackage{verbatim}
\usepackage{tabularx}
\usepackage{booktabs}
\usepackage{array}    
\usepackage{soul}
\usepackage{xcolor}

\usepackage{float}
\usepackage{amsmath}
\usepackage{amsthm}
\usepackage{amssymb}
\usepackage{graphicx}
\usepackage{subcaption}
\usepackage{authblk}
\usepackage{hyperref}
\usepackage{setspace}

\theoremstyle{plain}
\newtheorem{thm}{\protect\theoremname}
\theoremstyle{definition}

\theoremstyle{plain}

\theoremstyle{remark}
\newtheorem*{rem*}{\protect\remarkname}

\usepackage{tikz}
\providecommand{\definitionname}{Definition}
\providecommand{\propositionname}{Proposition}
\providecommand{\remarkname}{Remark}
\providecommand{\theoremname}{Theorem}

\title{\textbf{A Stage-Structured Deterministic Model of Fall Armyworm Infestation on Maize Farming}}

\author{ Donald Okoth Ojwang $^1$, Mamadou Pathe Ly $^{1,2}$,  Shaibu Osman $^3$}
\date{ %
	$^1$ Institute of Mathematical Sciences, African Institute for Mathematical Sciences, Senegal.\\%
	$^2$ Institut Ouest Africain de Mathematiques, Gamal Abdel Nasser University of Conakry, Guinea. \\
	$^3$ Department of Basic Sciences, University of Health and Allied Sciences, Ghana.\\[2ex]%
	{*} Corresponding author e-mail: ojwang.d.okoth@aims-senegal.org\\}
\newtheorem{case}{Case}
\begin{document}	
\maketitle	
\begin{abstract}
Fall Armyworm (FAW) poses a serious threat to maize production in many regions due to its aggressive feeding habits and rapid development cycle. In this study, we developed and analyzed a stage-structured mathematical model to evaluate the impact of different FAW larval instars on maize dynamics during the vegetative and reproductive stages. Analytical results indicate that the two models have unique and positively bounded solutions for all time $t\geq 0$ and admit four equilibrium points: the trivial, non-trivial, maize extinction and coexistence equilibria. The behavior of the model was studied using stability analysis to find conditions under which FAW dies out or continues to spread. Furthermore, sensitivity analysis and numerical simulations were conducted to examine how key parameters affect FAW population dynamics and maize. Numerical simulations of the model in both stages indicate that there is high destruction of maize plants in both vegetative and reproductive stages of maize production due to increased egg production and larval population density. The extensive damage caused by large populations of eggs, larvae and adult moths motivated an extension of the two models to include intervention measures such as traditional methods like handpicking and chemical pesticides. Numerical results indicate that these control strategies significantly suppressed the FAW population with a resultant increase in the maize plant population towards its maximum capacity during the vegetative and reproductive stages, respectively. 
\end{abstract}

\par analysis, Vegetative stage, Reproductive stage, larval instars, Maize plants.

\section{Introduction}
Agriculture plays a fundamental role in sustaining food security and economic development, especially in Sub-Saharan Africa, where a large proportion of the population depends on farming for their livelihoods. Among the major cereal crops, maize is one of the most widely cultivated and consumed staple foods, serving as a primary source of nutrition for millions of people \cite{adiaha2017impact} and was believed to have originated from Mexico about 7000 years ago and  transformed by Native Americans into a better source of food \cite{sharon2020severity}. Maize contributes significantly to household income, livestock feed production and industrial raw materials. In Sub-Saharan Africa, maize occupies more than 33 million hectares of cultivated land and remains a critical component of regional food security \cite{nwankwo2019boosting}.  Daudi et al. \cite{daudi2021dynamics} reported that 44 countries in Sub-Saharan Africa rely on white maize for 85–95\% of their staple food consumption.
Consequently, factors that threaten maize production have serious implications for agricultural sustainability, food availability and economic stability.

One of the most destructive threats to maize production is the Fall Armyworm (FAW), \textit{Spodoptera frugiperda}. FAW is a highly destructive \textit{lepidopteran polyphagous} pest, characterized by its strong migratory behavior, high reproductive potential and wide host range \cite{ghimire2021faw}. Since its first confirmed detection in West Africa in 2016 \cite{goergen2016first}, the pest has spread rapidly across the African continent, causing substantial yield losses and threatening the livelihoods of smallholder farmers. Studies have reported maize yield reductions of up to 90\% in heavily infested regions, contributing to food insecurity, malnutrition and poverty \cite{nget2024spodoptera}.  These outbreaks advanced over thousands of kilometers, reaching East African countries including Tanzania in January 2017, Kenya in April 2017 and Uganda in May 2017  \cite{daudi_modelling_2021}. 
Mathematical modelling has become an important tool for understanding complex biological systems and supporting decision making in agriculture. Differential equation models have been widely used to investigate pest population dynamics, crop-pest interactions and the effectiveness of control measures \cite{Bukova-Guzel2011}. In recent years, several mathematical models have been developed to study the impact of FAW infestation on maize farming. See \cite{daudi2021dynamics, daudi_modelling_2021, daudi2021mathematical, daudi2021fractional, angelov2016mathematical, Tchiengang2022multi, ali2021numerical, Dehingia2024,Dehingia2026,korobeinikov2009stability,Farman2026} for more details.

Daudi et al. \cite{daudi2021dynamics} developed a non-autonomous model using a Holling Type II functional response to describe FAW maize interactions in a periodic environment. They incorporated time-dependent control strategies, with simulations showing reduced FAW population under control.
Daudi et al. \cite{daudi2021mathematical} developed a stage-structured fractional-order model using the Caputo derivative to study FAW maize interactions, showing that pesticide use and handpicking reduced FAW population and increased maize biomass. On the other hand, Daudi et al. \cite{daudi_modelling_2021} proposed a two-subgeneric ODE model incorporating the developmental stages of maize and FAW, demonstrating that the immigration of adult moths increased maize biomass destruction, while harvesting and pesticide use improved yield. Anguelov et al. \cite{angelov2016mathematical} formulated an ODE model to assess mating disruption control using artificial female pheromones, showing that reduced fertilization lowered offspring production and male availability. Shabir et al. \cite{ali2021numerical} developed a fractional order FAW maize model using the Atangana–Baleanu fractional operator, incorporating all life stages and non-biological controls. Simulations showed faster stability for lower fractional orders, but the model overlooked larval cannibalism, which significantly affects population dynamics. Tchiengang et al. \cite{Tchiengang2022multi}  proposed a semi-discrete seasonal pest model and demonstrated the role of the basic reproduction number in determining pest persistence. Dehingia et al. \cite{Dehingia2026} proposed a Caputo fractional-order nutrient–plankton–fish interaction model to investigate ecosystem dynamics, stability and control of a complex biological system. Farman et al. \cite{Farman2026}
proposed a Caputo fractional-order model to study the interactions between global warming, dust pollution and plant biomass, highlighting the role of fractional calculus in analyzing stability and designing effective control strategies for environmental systems. Korobeinikov \cite{korobeinikov2009stability} developed a general predator–prey model and employed Lyapunov methods to establish global stability conditions for predator-free and coexistence equilibria, providing a theoretical basis for the analysis of interacting biological populations. Dehingia et al. \cite{Dehingia2024} developed a nutrient–phytoplankton–zooplankton model with refuge and time delay to investigate ecosystem dynamics. They showed that increasing the time delay can destabilize the system through Hopf bifurcation, while refuge temporarily stabilizes the population dynamics, highlighting the important role of time delays in ecological models.

The aforementioned studies, together with several others in the literature, have significantly advanced the understanding of FAW maize interactions. Nevertheless, important limitations remain. Most existing models treat all larval instars as a single homogeneous compartment, thereby neglecting differences in feeding intensity, developmental progression and crop damage among larval instars. This simplification may lead to an incomplete representation of the infestation process, since different larval instars contribute differently to maize damage. In addition, although larval cannibalism is a well documented ecological characteristic of FAW population, it is rarely incorporated explicitly into mathematical models. Ignoring cannibalistic interactions may underestimate their regulatory effect on larval survival and population growth \cite{assefa2019fall}. Consequently, the combined effects of larval stage heterogeneity and cannibalistic behaviour remain insufficiently explored in the current literature.

The novelty of this study can be summarized into two main points. First, unlike existing stage-structured FAW maize models that aggregate all larval instars into a single compartment, we partition the larval instars into early instar larvae (1-3) and late instar larvae (4-6), providing a more biologically realistic representation of larval development, feeding behaviour and crop damage. Second, the proposed model explicitly incorporates both intra-stage and inter-stage cannibalistic interactions through nonlinear terms. These extensions enable the model to capture important biological mechanisms that are often neglected in the literature, thereby providing a more realistic description of FAW population dynamics and their impact on maize.

The rest of the paper is organized as follows. Section \ref{sect2} presents the formulation of the stage-structured FAW maize model. Specifically, the model's steady states have been identified and their stability has been examined as well. Section \ref{sect3} formulates and analyzes the optimal control problem for managing FAW infestation. Section \ref{sect4} presents the numerical simulations and discusses the effects of the proposed control strategies. Finally, Section \ref{sect5}  concludes the paper and outlines directions for future research.

\section{Model Formulation and Analysis}\label{sect2}
\subsubsection{Model Formulation}
The model comprises two stage-structured populations, maize and FAW, leading to six distinct compartments. Maize growth is analyzed from emergence to full maturity over time  $t>0$ and is divided into two main periods, period I and period II. Period I occurs over the intervals $[0 , t_1]$, representing the vegetative stage. This stage includes key early growth events such as the planting of maize seeds, their emergence from the soil, whole leaf formation and the tasseling stage. Period II is assumed to take place over the time period $[t_1, t_2]$, which represents the reproductive stage and involves the development of corncob, silk, kernel and crop maturity.

The FAW population is divided into four compartments: eggs $E(t)$, early instar larvae $L_1(t)$ , late instar larvae $L_2(t)$ and adult moths $A(t)$. Although the biological life cycle of FAW includes a pupal stage, it is not modelled as a separate compartment in this study. Biologically, pupae do not feed and therefore do not contribute directly to maize damage \cite{Bista2020}. Its biological effect is incorporated into the effective maturation rate from larvae to adult moths. In particular, parameter $\delta_2$, reflects the combined duration of larval development and pupation, thereby accounting for the delay in adult emergence and its influence on population growth dynamics.

The egg population increases through oviposition by adult female moths at a constant rate  $\lambda$ and  decreases through hatching into larvae at rate $\beta$ and natural mortality at rate $\mu_1$. Similarly, the adult moth population grows as larvae mature into adult at a constant rate $\delta_2$ but declines over time due to natural mortality rate $\mu_4$. The model is formulated based on the following assumptions:
\begin{enumerate}
    \item  Maize is planted at time $t=0$ and each maize plant grows at a uniform and continuous rate from the vegetative to the reproductive stage.
   \item The maize plant population in a farm cannot exceed the carrying capacity $k$ as the crop approaches the maturity stage $T$.
    \item Maize is the only source of food for the larval instars, therefore, in the absence of maize, the larvae cannot survive and eventually die out.
    \item At time $t=0, M_1(0)=k$, where $k$ represents the field's carrying capacity.
    \item FAW infestation is at both the vegetative and reproductive stages of maize growth.
\end{enumerate}
\newpage
A summary of the  model parameters and variables is presented in Tables \ref{table1} and  \ref{table2} respectively.
\begin{table}[h!]
\centering
\caption{\textbf{Description of the model parameters.}}
\resizebox{\textwidth}{!}{%
\begin{tabular}{|c|m{16.5cm}|}
\hline
\textbf{Parameter} & \textbf{Description} \\
\hline
$b$ & Intrinsic growth rate of maize plant population at $t=0$ \\
\hline
$k$ & Maximum number of maize plants at $t=0$ \\
\hline
$\gamma_1$ & Maize destruction rate by early instar larvae in Period I \\
\hline
$\gamma_2$ & Maize destruction rate by late instar larvae in Period I \\
\hline
$\gamma_3$ & Maize destruction rate by early instar larvae in Period II \\
\hline
$\gamma_4$ & Maize destruction rate by late instar larvae in Period II \\
\hline
$\alpha$ & Proportion of female adult moths \\
\hline
$\lambda$ & Egg laying rate \\
\hline
$\beta$ & Hatching rate of FAW eggs into early instar larvae \\
\hline
$a_1$ & Conversion rate of maize into feed by early instar larvae (Period I) \\
\hline
$a_2$ & Conversion rate of maize into feed by late instar larvae  (Period I) \\
\hline
$a_3$ & Conversion rate of maize into feed by early instar larvae (Period II) \\
\hline
$a_4$ & Conversion rate of maize into feed by late instar larvae (Period II) \\
\hline
$\delta_1$ & Transition rate of early to late instar larvae \\
\hline
$P_1$ & Rate of cannibalism between early and late instar larvae  \\
\hline
$P_2$ & Rate of cannibalism within late instar larvae  \\
\hline
$\mu_1$ & Mortality rate of egg \\
\hline
$\mu_2$ & Natural mortality rate of early instar larvae \\
\hline
$\mu_3$ & Natural mortality rate of late instar larvae \\
\hline
$\mu_4$ & Natural mortality rate of adult moth \\
\hline
$\delta_2$ & Rate at which late instar larvae develop into adult moth \\
\hline
\end{tabular}}
\label{table1}
\end{table}

\begin{table}[h!]
\centering
\caption{\textbf{Description of the model variables.}}
\begin{tabular}{|c|p{14.5cm}|}
\hline
\textbf{Variables} & \textbf{Description} \\
\hline
$M_1(t)$ & Maize plant population in the vegetative stage at any time t \\
\hline
$M_2(t)$ & Maize plant population in the reproductive stage at any time t \\
\hline
$E(t)$ & Egg population density at any time t \\
\hline
$L_1(t)$ & Early instar larvae population at time $t > 0$ \\
\hline
$L_2(t)$ & Late instar larvae population at time $t > 0$ \\
\hline
$A(t)$ & Population density of adult moth at any time $t$ \\
\hline
\end{tabular}
\label{table2}
\end{table}
All state variables and model parameters are assumed to be non-negative, consistent with their biological and ecological interpretation. The compartmental dynamics are described in detail for the vegetative stage. Since the equations governing the FAW life stages in the reproductive stage follow the same formulation, only the maize equation is presented separately, while the same biological interpretation applies to the remaining state variables.
\\
The dynamics of the maize plant population is governed by the equation
\begin{align}
\frac{dM_1}{dt} = bM_1 \left(1 - \frac{M_1}{k} \right) - (\gamma_1 L_1 + \gamma_2 L_2) M_1,\tag{1}
\end{align}
where $bM_1\left(1-\frac{M_1}{k}\right)$ represents the logistic growth of maize with carrying capacity $k$. The terms $\gamma_1L_1M_1$ and $\gamma_2L_2M_1$ represent the reduction in the maize plant population due to feeding by the early instar larvae and late instar larvae, respectively. Similarly, maize dynamics during the reproductive stage is governed by
\begin{align}
\frac{dM_2}{dt} = b M_2 \left(1 - \frac{M_2}{k} \right) - (\gamma_3 L_1 + \gamma_4 L_2) M_2,\tag{2}
\end{align}
where $bM_2\left(1-\frac{M_2}{k}\right)$ represents the logistic growth of maize with carrying capacity $k$. The terms \( \gamma_3 L_1 M_2 \) and \( \gamma_4 L_2 M_2 \) represent maize destruction by the early instar larvae and late instar larvae, respectively.  The following equation describes the FAW egg population dynamics:
\begin{align}
\frac{dE}{dt} = \alpha \lambda A - (\beta + \mu_1) E,\tag{3}
\end{align}
where the term $\alpha\lambda A$ represents egg production by female adult moths. The egg population decreases through hatching into early instar larvae at rate $\beta E$  and  natural mortality at rate  $\mu_1 E$. The dynamics of the early instar larvae is described by the equation
\begin{align}
\frac{dL_1}{dt} = \beta E + a_1 \gamma_1 L_1 M_1 - \delta_1 L_1 - P_1 L_1 L_2 - \mu_2 L_1.\tag{4}
\end{align}
The early instar larvae increase through egg hatching and feeding on maize plants, represented by the terms $\beta E$ and $a_1\gamma_1L_1M_1$, respectively. It decreases through development into the late instar larvae at rate $\delta_1 L_1$, inter-stage cannibalism at rate $P_1 L_1 L_2$ and natural mortality at rate $\mu_2 L_1$. The dynamics of the late instar larvae is given by
\[
\frac{dL_2}{dt} = \delta_1 L_1 + a_2 \gamma_2 L_2 M_1 + P_1 L_1 L_2 - P_2 L_2^2 - \delta_2 L_2 - \mu_3 L_2.\tag{5}
\]
The late instar larvae increase through the maturation of early instars, maize consumption and inter-stage cannibalism, represented by the terms $\delta_1 L_1$, $a_2 \gamma_2 L_2 M_1$ and $P_1 L_1 L_2$, respectively. It decreases through intra-stage cannibalism, progression into adult moth and natural mortality, represented by $P_2 L_2^2$, $\delta_2 L_2$ and $\mu_3 L_2$, respectively. The following equation represents the dynamics of adult moth stage:
\[
\frac{dA}{dt} = \delta_2 L_2 - \mu_4 A,\tag{6}
\]
which increases through the maturation of late instar larvae, represented by the term $\delta_2 L_2$ and decreases through natural mortality represented by $\mu_4A$. The model explanations above can be represented schematically as shown in Figure \ref{fig:1} below.
\begin{figure}[H]
    \centering
    \includegraphics[width=0.70\linewidth]{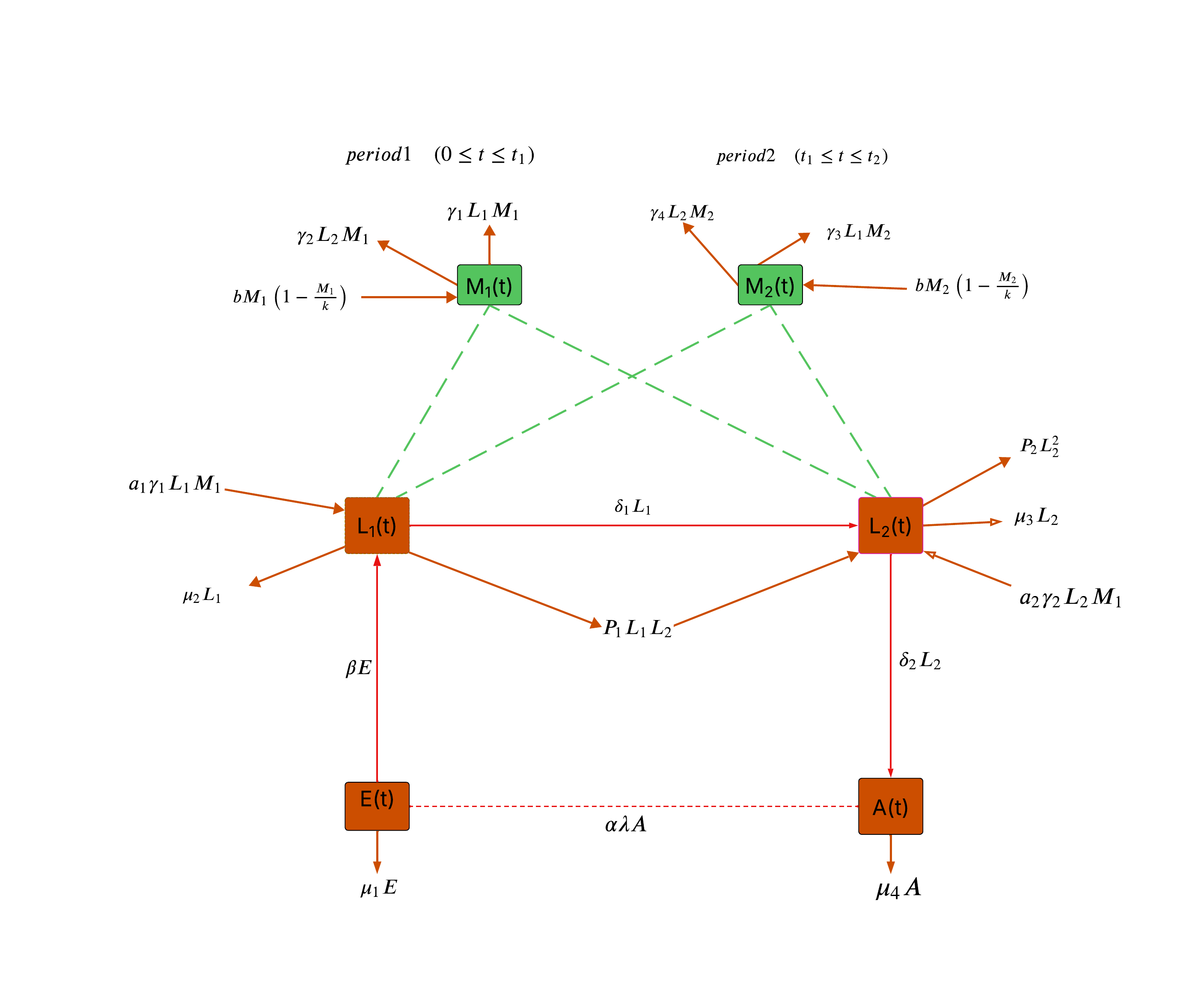}
    \caption{\textbf{Compartmental flow diagram illustrating the stage-structured dynamics of FAW infestation in a maize plantation.}}
    \label{fig:1}
\end{figure}
\subsubsection{Model Equations}
Based on the assumptions provided, the mathematical models governing the vegetative and reproductive stages of maize growth are presented in systems \eqref{equation7} and \eqref{equation8}, respectively and illustrated in Figure \ref{fig:1}. Since the same maize plants that survive in the vegetative stage continue into the reproductive stage, we assume that the maize plant population is continuous at the transition time $t=t_1$. Therefore, the final maize plant population in the vegetative stage serves as the initial maize plant population in the reproductive stage, that is, $M_1(t_1)=M_2(t_1)$. 
Similarly, the terminal values of the egg, early instar larvae, late instar larvae and adult moth populations at the end of the vegetative stage are used as the initial conditions for the reproductive stage, that is, $E(t_1^-)=E(t_1^+),$
$L_1(t_1^-)=L_1(t_1^+),$
$L_2(t_1^-)=L_2(t_1^+),$
$A(t_1^-)=A(t_1^+)$. This assumption reflects the biological interpretation that maize plant population and the FAW population do not disappear and reappear at the stage transition but rather progress continuously from the vegetative growth to the reproductive development. The model in the vegetative stage ($0 \leq t \leq t_1$) is given by

\begin{equation}\label{equation7}
\left\{
\begin{aligned}
\frac{dM_1}{dt} &= bM_1 \left(1 - \frac{M_1}{k} \right) - (\gamma_1 L_1 + \gamma_2 L_2)M_1 ,
\\
\frac{dE}{dt} &= \alpha \lambda A - (\beta + \mu_1) E, 
\\
\frac{dL_1}{dt} &= \beta E + a_1 \gamma_1 L_1 M_1 - \delta_1 L_1 - P_1 L_1 L_2 - \mu_2 L_1,
\\
\frac{dL_2}{dt} &= \delta_1 L_1 + a_2 \gamma_2 L_2 M_1 + P_1 L_1 L_2 - P_2 L_2^2 - \delta_2 L_2 - \mu_3 L_2,
\\
\frac{dA}{dt} &= \delta_2 L_2 - \mu_4 A,
\end{aligned}
\right. \tag{7}
\end{equation}
with initial conditions,
\begin{align*}
M_1(0) = k,\hspace{0.5cm}
E(0) \geq 0, \hspace{0.5cm} 
L_1(0) \geq 0, \hspace{0.5cm} 
L_2(0) \geq 0, \hspace{0.5cm}
A(0) \geq 0 .\hspace{0.5cm} 
\end{align*}
The model in the reproductive stage ($t_1 \leq t \leq t_2$) is 
\begin{equation}\label{equation8}
\left\{
\begin{aligned}
\frac{dM_2}{dt} &= bM_2 \left(1 - \frac{M_2}{k} \right) - (\gamma_3 L_1 + \gamma_4 L_2)M_2, 
\\
\frac{dE}{dt} &= \alpha \lambda A - (\beta + \mu_1) E, 
\\
\frac{dL_1}{dt} &= \beta E + a_3 \gamma_3 L_1 M_2 - \delta_1 L_1 - P_1 L_1 L_2- \mu_2 L_1,
\\
\frac{dL_2}{dt} &= \delta_1 L_1 + a_4 \gamma_4 L_2 M_2 + P_1 L_1 L_2 - P_2 L_2^2 - \delta_2 L_2 - \mu_3 L_2 ,
\\
\frac{dA}{dt} &= \delta_2 L_2 - \mu_4 A,
\end{aligned}
\right. \tag{8}
\end{equation}
with initial conditions,
$M_1(t_1) = M_2(t_1)$  and
$M_2(t) = 0$ for  $t < t_1.$

\subsection{Model Analysis}
\subsubsection{Positivity of Solutions}
For systems \eqref{equation7} and \eqref{equation8} to be biologically meaningful and  mathematically well posed, it is essential to prove that all state variables remain positive for all $t>0$ under their respective initial conditions. This property is established by the following theorem.
\begin{thm}
All solutions of systems \eqref{equation7} and \eqref{equation8} in all two time periods $t>0$ are non-negative and the solution remains in $\mathbb R^5_+$.
\end{thm}
\begin{proof}
To prove this theorem, we let
\begin{align*}
\Omega_1 =\left\{ (M_1,E,L_1,L_2,A)\in \mathbb R^5_+ : M_1(0)=k,E(0)\geq0,L_1(0)\geq0,L_2(0)\geq0,A(0)\geq0\right\}
,\end{align*}
for system \eqref{equation7} and 
\begin{align*}
\Omega_2 =\left\{ (M_2,E,L_1,L_2,A)\in \mathbb R^5_+ : M_2(0)=0,E(0)\geq0,L_1(0)\geq0,L_2(0)\geq0,A(0)\geq0\right\},
\end{align*}
for system \eqref{equation8}. Then, the solution sets $M_1,E,L_1,L_2,A$ and $M_2,E,L_1,L_2,A$ for systems \eqref{equation7} and \eqref{equation8}, respectively are positive for all $t>0$ in the interval $[0,t_2]$.
\begin{case}\label{case 1}
We first establish the positivity of solutions for system \eqref{equation7}, which describes the vegetative stage of maize growth.
\end{case}
From the maize equation, we have
\[
\frac{dM_1}{dt} = bM_1 \left(1 - \frac{M_1}{k} \right) - (\gamma_1 L_1 + \gamma_2 L_2) M_1. 
\] 
Since  \(-(\gamma_1 L_1 + \gamma_2 L_2) M_1 \leq 0
\), it follows that $\frac{dM_1}{dt} \leq bM_1 \left(1 - \frac{M_1}{k} \right)$. Rewriting in terms of $V=\frac{1}{M_1}$ yields,
$V' + bV \geq \frac{b}{k}.$
Solving this inequality results in
\[
M_1(t) \leq \frac{k^2 M_1(0)}{e^{-bt}(k - M_1(0)) + kM_1(0)}.
\]
Hence, as \( t \to \infty \), we obtain \( 0 \leq M_1(t) \leq k \), showing that the maize plant population cannot exceed the carrying capacity $k$.
Similarly, we consider the egg equation
\[
\frac{dE}{dt} = \alpha \lambda A - (\beta  + \mu_1) E.
\]
Since \(\alpha \lambda A \geq 0\), it follows that
$\frac{dE}{dt} \geq -(\beta + \mu_1)E.$ Separating the variables and integrating results in
\[
E(t) \geq C_2 e^{-(\beta + \mu_1)t}.
\]
For the early instar larvae equation, we have
\[
\frac{dL_1}{dt} = \beta E + a_1 \gamma_1 L_1 M_1 - \delta_1 L_1 - P_1L_1 L_2 - \mu_2 L_1.
\]
Since \((\beta E + a_1 \gamma_1 L_1 M_1) \geq 0\), it follows that $\frac{dL_1}{dt} \geq -(\delta_1 + P_1L_2 + \mu_2) L_1.$
Separating the variables gives
\[
\frac{dL_1}{L_1} \geq -(\delta_1 + P_1L_2 + \mu_2) \ dt.
\]
Integrating both sides leads to
\[
 L_1(t) \geq L_1(0) e^{-(\delta_1 + P_1 L_2 + \mu_2)t}.
\]
Likewise, for the late instar larvae equation,
\[
\frac{dL_2}{dt} = \delta_1 L_1 + a_2 \gamma_2 L_2 M_1 + P_1 L_1 L_2 - P_2 L_2^2 - \delta_2 L_2 -\mu_2 L_2.
\]
Since
\(
\delta_1 L_1 + a_2 \gamma_2 L_2 M_1 + P_1 L_1 L_2  \geq 0,
\) it follows that
\[
\frac{dL_2}{dt} + ( \delta_2 +\mu_3 ) L_2 \geq -P_2 L_2^2.
\]
Integrating both sides results in
\[
L_2(t) \geq \frac{L_2(0)(\delta_2 + \mu_3)}{(\delta_2 + \mu_3)\left[(\delta_2 + \mu_3) + P_2 L_2(0)\right] e^{(\delta_2 + \mu_3)t} - L_2(0) P_2}.
\]
Lastly, let us consider the adult moth equation
\[
\frac{dA}{dt} = \delta_2 L_2 - \mu_4 A.
\]
Since \(\delta_2 L_2 \geq 0\), it follows tha
$\frac{dA}{dt} \geq -\mu_4 A.$
Separating the variables and integrating both sides results in
\[
A(t) \geq C_5 e^{-\mu_4 t}.
\]
At \(t = 0\),
\[
A(0) \geq C_5 \Longrightarrow A(t) \geq A(0) e^{-\mu_4 t} \geq 0.
\]
Therefore, all equations in system \eqref{equation7} for $0 \leq t \leq t_1$ are non-negative. Thus the solutions of sets  $\left\{ (M_1,E,L_1,L_2,A)\right\}$
for $t\geq0$ exist in $\Omega_1$ and the state variables of system \eqref{equation7} are non-negative for all time $t>0$.

Applying the same approach used in Case \eqref{case 1},  we can prove that all equations in system (\ref{equation8}) satisfy the non-negativity condition, resulting in \(M_2 \geq 0\), \(E \geq 0, L_1 \geq 0, L_2 \geq 0, A \geq 0\) which exist in \(\Omega_2\). This confirms that the equations representing the reproductive stage are also positive.
\end{proof}

\subsubsection{Invariant Region}
This section focuses on finding a region where  the behaviors of systems \eqref{equation7} and \eqref{equation8} stay within limits and do not grow without bounds. We first establish the invariant region for system \eqref{equation7}, which represents the vegetative stage of maize growth. 

\begin{thm}
Let $Y(t)=\big(M_1(t),E(t),L_1(t),L_2(t),A(t)\big)$
be a unique solution of system \eqref{equation7} with non-negative initial conditions. Then \(Y(t)\) is bounded for all \(t\geq0\). That is, $Y(t)\in\Omega$, where the feasible region \(\Omega\) is defined by
\[
\Omega = \left\{ (M_1, E, L_1, L_2, A) \in \mathbb{R}_+^5 :
\begin{array}{l}
0 \leq M_1(t) \leq k,\; 
0 \leq E(t) \leq E_{\max},\; 
0 \leq L_1(t) \leq L_{1,\max},\; 
0 \leq L_2(t) \leq L_{2,\max}, \\
0 \leq A(t) \leq A_{\max}
\end{array}
\right\}.
\]
\end{thm}
\begin{proof}
    We establish the boundedness of each state variable by analyzing the differential equations.
    From the governing equation for the maize plant population and the comparison argument established in the positivity analysis, it follows that
\[
0 \leq M_1(t) \leq k, \qquad \text{for all } t \geq 0.
\]
Therefore, the maize plant population remains bounded above by the carrying capacity \(k\). Thus,
\[
\left\{ M_1(t) \in \mathbb{R}_+ : 0 \leq M_1(t) \leq k \right\}.
\]
From the egg population, we have
\[
\frac{dE}{dt} = \alpha \lambda A - (\beta + \mu_1)E.
\]
Assuming $A(t)\leq A_{max}$, the egg equation satisfies
\[
\frac{dE}{dt} \leq \alpha \lambda A_{max} - (\beta + \mu_1)E.
\]
Solving the corresponding linear differential inequality yields an upper bound for the egg population given by
\[
\left\{ E(t) \in \mathbb{R}_+ : 0 \leq E(t) \leq \frac{\alpha \lambda A_{\max}}{\beta + \mu_1} \right\},
\]
where $E_{max}=\frac{\alpha \lambda A_{\max}}{\beta + \mu_1}$.
\\
For the larval instars, since \( L_1 \) and \( L_2 \) represent the same biological class, we define the total larval population as
\begin{align*}
    L &= L_1 + L_2, \\
    \frac{dL}{dt} &= \beta E + a_1 \gamma_1 L_1 M_1 -\delta_1 L_1 - P_1 L_1 L_2 - \mu_2 L_1 + \delta_1 L_1 + a_2 \gamma_2 L_2 M_1 + P_1L_1L_2 - P_2L^{2}_2 - \delta_2 L_2 - \mu_3L_2,\\
 \frac{dL}{dt} &= \beta E + a_1 \gamma_1 L_1 M_1 + a_2 \gamma_2 L_2 M_1 - (\mu_2 L_1 + P_2 L^{2}_2 + \delta_2 L_2 + \mu_3 L_2).
\end{align*}
We let $ q_2 = \max \{ M_1 (0), k \} $, $ h_1 = \min \{ \mu_2, P_2, \delta_2, \mu_3 \} $ and 
$ h_2 = \max \{ h(E(0)), f \} .$
Where \( M_1 \leq q_2 \) and \( E \leq h_2 \). It follows that
\[
\frac{dL}{dt} \leq \beta E + a_1 \gamma_1 L_1 k + a_2 \gamma_2 L_2 k - h_1 L,
\]
\[
\frac{dL}{dt} \leq \beta f + (m_1 L_1 + m_2 L_2) - h_1 L,
\]
where $m_1 = a_1 \gamma_1 k$, $m_2 = a_2 \gamma_2 k$ and \( L = \max \{ L_1, L_2 \} \). This further simplifies to
\[
\frac{dL}{dt} \leq \beta f + (m - h_1) L,
\]
\[
\frac{dL}{dt} \leq \beta f + h_3 L,
\]
where $m = m_1 + m_2 $ and $h_3 = m - h_1$.
Introducing the integrating factor and solving yields
\[
L(t) \leq \frac{\beta f}{h_3} + \left( L(0) + \frac{\beta f}{h_3} \right) e^{-h_3 t}.
\]
As \( t \to \infty \), the larval instars satisfies the upper bound
\[
\left\{ L_1(t), L_2(t) \in \mathbb{R}_+^2 : 0 \leq L(t) \leq \frac{\beta f}{h_3} \right\},
\]
where $ L_1,_{max}=\frac{\beta f}{h_3}$ and $L_2,_{max}=\frac{\beta f}{h_3}$.
\\
From the  adult moth population, we have
\[
\frac{dA}{dt} = \delta_2 L_2 - \mu_4 A.
\]
We let \( L_2(t) \leq L_{2,\max} \), where \( L_{2,\max} \) denotes the maximum attainable late instar larval population. It follows that
\[
\frac{dA}{dt} + \mu_4 A \leq \delta_2 L_{2,\max}.
\]
Integrating and solving results in the upper bound
\[
\left\{ A(t) \in \mathbb{R}_+ : 0 \leq A(t) \leq \frac{\delta_2 L_{2,\max}}{\mu_4} \right\},
\]
where $A_{max}=\frac{\delta_2 L_{2,\max}}{\mu_4}$.
\\
Using the same technique,  we can prove that the equations in system \eqref{equation8} for $t_1 \leq t \leq t_2$, with their respective initial conditions have bounded non-negative solutions. This completes the proof.
\end{proof}
\subsubsection{Basic Reproduction Number $(R_0)$}
In ecological predator–prey systems, the basic reproduction number represents the threshold quantity that determines whether a predator population can successfully establish and persist in the presence of its prey \cite{Liu2018Age}. In the proposed FAW maize model, the larval instars act as the feeding population responsible for crop damage, while maize serves as the resource supporting larval development. Consequently, the basic reproduction number measures the ability of the larval instars to maintain positive growth through feeding relative to losses due to natural mortality. Following the ecological predator--prey framework of Korobeinikov \cite{korobeinikov2009stability}, the basic reproduction number is defined as the ratio of predator population gain through feeding to predator loss through natural mortality, both evaluated at the predator-free equilibrium, ($k,0,0,0,0$). Accordingly,
\begin{align}
R_0 = \frac{K \frac{\partial \omega (x_1^0, x_2^0)}{\partial x_2}}{\frac{d\mu(x^0_2)}{dx_2}},\tag{9}
\end{align}
where $K$ represents the feed conversion rate, $\partial \omega (x_1^0, x_2^0)$ represents the prey destruction rate and $d\mu(x^0_2)$ is the mortality rate of the predator in the absence of prey. In the present model, the feed conversion rates are represented by the parameters $a_1$ and $a_2$, corresponding to the early and late instar larvae, respectively. Consequently, the contributions to the basic reproduction number arise from  the early instar larvae and late instar larvae and are computed separately as follows,
\[
R_0 = \frac{a_1 \frac{\partial u_1(M_1^0, L_1^0)}{\partial L_1}}{\frac{d\Phi_1(L^0_1,L^0_2)}{dL_1}} + \frac{a_2 \frac{\partial u_2(M_1^0, L_2^0)}{\partial L_2}}{\frac{d\Phi_2(L^0_1,L^0_2)}{dL_2}}. \tag{10}
\]
The contributions of the early instar larvae and late instar larvae to the  basic reproduction number are computed separately. For the early instar larvae,
\[
R_{0_1} = \frac{a_1 \gamma_1 k}{\mu_2},\tag{11}
\]
while for the late instar larvae,
\[
R_{0_2} = \frac{a_2 \gamma_2 k}{\mu_3}.\tag{12}
\]
Therefore, the overall basic reproduction number is
\[
R_0 = R_{0_1} + R_{0_2} = \frac{a_1 \gamma_1 k}{\mu_2} + \frac{a_2 \gamma_2 k}{\mu_3}.\tag{13}
\]
This expression shows that the  basic reproduction number increases with the larval feeding rates $a_1$ and $a_2$, feed conversion rate $k$ and maize destruction rates $\gamma_1$ and $\gamma_2$, but decreases with the natural mortality rates of the early instar larvae and late instar larvae $\mu_2$ and $\mu_3$, respectively. 
In particular, $R_0<1$ means that larval growth is too low to overcome natural mortality, causing the FAW population to eventually die out. Conversely, $R_0>1$ means that larval growth exceeds natural mortality, allowing the FAW population to survive, establish itself and continue damaging maize plants.
\subsubsection{ Equilibrium Points and their Existence}\label{section2.1.4}
Solving system \eqref{equation7} at steady state yields the following four equilibrium points:\\
\\
\textbf{(a)} The trivial equilibrium point
$\mathcal{E}_1 : (M_1^0, E^0, L_1^0, L_2^0, A^0) = (0,0,0,0,0)$ always exists and represents the complete absence of both maize plants and the FAW population. Biologically, it corresponds to an empty system with neither crop nor pest present.
\\
\\
\textbf{(b)} The non trivial equilibrium point
$\mathcal{E}_2:(M_1^1, E^1, L_1^1, L_2^1, A^1)=(k,0,0,0,0)$ always exists and represents a state where the maize plant population persists while the FAW population is eliminated. Biologically, it indicates successful control of the pest and sustainable maize production.
\\
\\
\textbf{(c)} The maize extinction equilibrium point $\mathcal{E}_3$ \eqref{equation19} represents a state where the maize plant population is completely depleted while the FAW population persists. Biologically, it corresponds to severe infestation resulting in crop failure. It is obtained by setting all derivatives equal to zero and substituting $M_1^3 = 0$ in system \eqref{equation7} to obtain
\begin{equation}\label{equation14}
\left\{
\begin{aligned}
0 &= M_1^3 \\
0 &= \alpha \lambda A^3 - (\beta + \mu_1) E^3 \\
0 &=\beta E^3 - (P_1 L_2^3 + (\delta_1 + \mu_2)) L_1^3 \\
0 &= \delta_1 L_1^3 + (P_1 L_1^3 - P_2 L_2^{3} - (\delta_2 + \mu_3 ))L_2^3 \\
0 &= \delta_2 L_2^3 - \mu_4 A^3.
\end{aligned}
\right. \tag{14}
\end{equation}
We then let $\omega_1 = \alpha \lambda$, $\omega_2 = \beta + \mu_1$, $\omega_3 = \delta_1 + \mu_2$ and $\omega_4 = \delta_2 + \mu_3$. Solving the second and fifth equations of \eqref{equation14} and substituting the resulting expression for $A^3$ into $E^3$ yields
\begin{equation}\label{equation15}
E^3 =  qL_2^3, \tag{15}
\end{equation}
where \( q = \frac{\omega_1 \delta_2}{\omega_2 \mu_4} \).
Substituting equation \eqref{equation15} in the third equation of system  \eqref{equation14} gives
\begin{equation}\label{equation16}
0 = \beta q L_2^3 - \left(P_1 L_2^3 + \omega_3 \right) L_1^3. \tag{16}
\end{equation}
Solving the fourth equation of system \eqref{equation14} for $L_1^{3}$, substituting the resulting expression into equation \eqref{equation16} and simplifying yields the quadratic equation
\[ \tag{17}\label{eq17}
P_1P_2 (L_2^{3})^2 + (P_1 \omega_4 + P_2 \omega_3- \beta q P_1 ) L_2^3  + (\omega_3\omega_4 - \beta q \delta_1)= 0.
\]
Letting
$y_1 = P_1 P_2$, $y_2 = P_1 \omega_4 + P_2 \omega_3 - \beta q P_1$ \text{and} $y_3 = \omega_3 \omega_4 - \beta q\delta_1$ and applying quadratic formula to equation \eqref{eq17} yields
\[ \tag{18}\label{eq18}
L_2^{3} = \frac{-y_2 \pm \sqrt{y_2^2 - 4y_1 y_3}}{2y_1}.
\]
Substituting \eqref{eq18} into the expressions for $E^3$, $L_1^{3}$, $L_2^{3}$ and $A^3$ and replacing with the conventions of $y_1,y_2,y_3$, yields the maize extinction equilibrium point given by

{\tiny
\begin{equation}\label{equation19}
\left\{
\begin{aligned}
M_1^3 &= 0,\\
\\
E^3 &= \alpha\lambda\delta_2\left\{ \frac{-P_1(\delta_2 + \mu_3) - P_2(\delta_1 + \mu_2) + \frac{\beta \alpha \lambda \delta_2}{(\beta + \mu_1)\mu_4} P_1 \pm \sqrt{\left( P_1(\delta_2 + \mu_3) + P_2(\delta_1 + \mu_2) - \frac{\beta \alpha\lambda \delta_2}{(\beta + \mu_1)\mu_4} P_1 \right)^2 - 4P_1P_2 [(\delta_1 + \mu_2)(\delta_2 + \mu_3) - \frac{\beta \alpha\lambda\delta_1 \delta_2}{(\beta + \mu_1)\mu_4}}] }{2(\beta + \mu_1)\mu_4P_1 P_2} \right\}, \\[5pt]
\\
L_1^3 &= \left\{ \frac{ \left( \frac{P_2}{4y^2_1} \left( -y_2 \pm \sqrt{y_2^2 - 4 y_1 y_3} \right) + \omega_4) \left( -y_2 \pm \sqrt{y_2^2 - 4 y_1 y_3} \right) \right) }{ 2 y_1 \left( \delta_1 + P_1 \left( -y_2 \pm \sqrt{y_2^2 - 4 y_1 y_3} \right) \right) }\right\} ,\\[5pt]
\\
L_2^3 &= \left\{ \frac{-P_1(\delta_2 + \mu_3) - P_2(\delta_1 + \mu_2) + \frac{\beta\alpha\lambda \delta_2}{(\beta + \mu_1)\mu_4} P_1 \pm \sqrt{\left( P_1(\delta_2 + \mu_3) + P_2(\delta_1 + \mu_2) - \frac{\beta \alpha\lambda\delta_2}{(\beta + \mu_1)\mu_4} P_1 \right)^2 - 4P_1P_2[(\delta_1 + \mu_2)(\delta_2 + \mu_3) - \frac{\beta \alpha\lambda\delta_1\delta_2}{(\beta + \mu_1)\mu_4}}] }{2P_1 P_2} \right\}, \\[5pt]
\\
A^3 &= \left\{ \frac{\delta_2 (-P_1(\delta_2 + \mu_3) - P_2(\delta_1 + \mu_2) + \frac{\beta \alpha\lambda \delta_2}{(\beta + \mu_1)\mu_4} P_1 \pm \sqrt{\left( P_1(\delta_2 + \mu_3) + P_2(\delta_1 + \mu_2) - \frac{\beta \alpha\lambda \delta_2}{(\beta + \mu_1)\mu_4} P_1 \right)^2 - 4P_1P_2[(\delta_1 + \mu_2)(\delta_2 + \mu_3) - \frac{\beta \alpha\lambda\delta_1 \delta_2}{(\beta + \mu_1)\mu_4}}]}{2P_1 P_2 \mu_4} \right\} .\\[5pt]
\end{aligned}
\right.\tag{19}
\end{equation}
}
Therefore, the equilibrium point $\mathcal{E}_3$ \eqref{equation19} exists and is biologically feasible if   ($P_1(\delta_2 + \mu_3) + P_2(\delta_1 + \mu_2) - \frac{\beta \alpha\lambda\delta_2}{(\beta + \mu_1)\mu_4} P_1 )^2 - 4P_1P_2(\delta_2 + \mu_2)(\delta_2 + \mu_3) - \frac{\beta \alpha\lambda\delta_1\delta_2}{(\beta + \mu_1)\mu_4}>0$ and $P_1P_2>0.$\\
\\
\\
\textbf{(d)} The coexistence equilibrium point $\mathcal{E}_4$ \eqref{eq23} represents a state where both the maize plant population and the FAW population persist at constant levels. Biologically, it reflects a balance between maize growth and FAW infestation. It is obtained by setting all derivatives in system \eqref{equation7} equal to zero to obtain
\begin{equation}\label{equation20}
\left\{
\begin{aligned}
0& = b\left(1 - \frac{M_1^4}{k}\right) - (\gamma_1 L_1^4 + \gamma_2 L_2^4) \\
0& = \omega_1 A^4 - \omega_2 E^4 \\
0 &= \beta E^4 + (a_1 \gamma_1 M_1^4 - P_1 L_2^4 - \omega_3) L_1^4 \\
0& = \delta_1 L_1^4 + (a_2 \gamma_2 M_1^4 + P_1 L_1^4 - P_2 L_2^4 - \omega_4) L_2^4\\ 
0& = \delta_2 L_2^4 - \mu_4 A^4.
\end{aligned}
\right.\tag{20}
\end{equation}
Recalling  equation \eqref{equation15} and substituting the expression for $M_1^{4}$ into the third equation of \eqref{equation20} gives
\begin{equation}\label{equation21}
  L_2^4 = \frac{Z_1 (L_1^{4})^2 - Z_2 L_1^4}{Z_3 - Z_4 L_1^4}. \tag{21}
\end{equation}
Expressing $M_1^{4}$ from the first equation of \eqref{equation20}, substituting into its fourth equation and solving gives
\begin{equation}\label{equation22}
     0=x_1L^4_1+(x_2+x_3L^4_1-x_4L^{4}_2)L^4_2.\tag{22}
\end{equation}
Substituting the expression for \eqref{equation21} into  equation \eqref{equation22} and simplifying gives a cubic equation 
\begin{align*}
L_1^{4} =\ & \sqrt[3]{\frac{-\left(-\frac{\Gamma_3^3}{27} + \frac{\Gamma_3^3}{9} - \frac{\Gamma_3 \Gamma_2}{3} + \Gamma_1\right) + \sqrt{\left(-\frac{\Gamma_3^3}{27} + \frac{\Gamma_3^3}{9} - \frac{\Gamma_3 \Gamma_2}{3} + \Gamma_1\right)^2 + \frac{4}{27}\left(\Gamma_2 - \frac{\Gamma_3^2}{3}\right)^3}}{2}} \\
& - \sqrt[3]{\frac{-\left(-\frac{\Gamma_3^3}{27} + \frac{\Gamma_3^3}{9} - \frac{\Gamma_3 \Gamma_2}{3} + \Gamma_1\right) - \sqrt{\left(-\frac{\Gamma_3^3}{27} + \frac{\Gamma_3^3}{9} - \frac{\Gamma_3 \Gamma_2}{3} + \Gamma_1\right)^2 + \frac{4}{27}\left(\Gamma_2 - \frac{\Gamma_3^2}{3}\right)^3}}{2}} - \frac{\Gamma_3}{3}.
\end{align*}
Where
$\Gamma_3 = \frac{m_2}{m_1},$ 
$\Gamma_2 = \frac{m_3}{m_1},$ 
$\Gamma_1 = \frac{m_4}{m_1}$ and $m_1 = x_1 z_3^2 - x_2 x_3 z_2$, $ m_2 = x_2 z_4 z_2 - x_2 z_3 x_3 - 2x_1 z_3 z_4 - x_4 z_2^{2},$\\
$m_3 = x_1 z_4^{2} + x_3 z_3 z_1 - x_2 z_4 z_1 + x_3 z_4 z_2 + 2x_4 z_1 z_2,$ $m_4 = x_3 x_4 z_1 + x_4 z_1^{2},$ $x_1 = \delta_1b ,$ $x_2 = a_2 \gamma_2 k b - \omega_4 b,$ \\
$x_3 = P_1 b - a_2 k \gamma_2 \gamma_1 ,$ $x_4 = P_2 b + a_2k \gamma_2^2, Z_1 = a_1k \gamma_1^{2},$ 
$Z_2 = a_1 \gamma_1 kb - \omega_3 b,$ 
$Z_3 = \beta q b,$ 
$Z_4 = a_1 \gamma_1\gamma_2 k + P_1 b.$
Let $v_1$ denote its biologically feasible root. Substituting $L_1^{4}=v_1$ into the corresponding expressions for $M_1^4$, $E^4$, $L_1^4$, $L_2^4$ and $A^4$ yields the coexistence equilibrium point 
 
\begin{equation}\label{eq23}
\left\{
\begin{aligned}
M_1^4 &= \frac{kb - k\left(\gamma_1 v_1 + \gamma_2\left(\frac{a_1k \gamma_1^2 v_1^2 - (a_1 \gamma_1 kb - \omega_3 b) v_1}{\beta q b - (a_1 \gamma_1 \gamma_2 k + P_1 b) v_1} \right) \right)}{b} ,\\
E^4 &= \frac{q a_1 k \gamma_1^2 v_1^2 - q v_1 (a_1 \gamma_1 kb - \omega_3 b)}{\beta q b - (a_1 \gamma_1 \gamma_2 k + P_1 b) v_1}, \\
L_1^4 &= v_1 ,\\
L_2^4 &= \frac{a_1 k \gamma_1^2 v_1^2 - v_1 (a_1 \gamma_1 kb - \omega_3 b)}{\beta q b - (a_1 \gamma_1 \gamma_2 k + P_1 b) v_1} ,\\
A^4 &= \frac{\delta_2 a_1 k \gamma_1^2 v_1^2 - \delta_2 v_1 (a_1 \gamma_1 kb - \omega_3 b)}{[\beta q b - (a_1 \gamma_1 \gamma_2 k + P_1 b) v_1] \mu_4}.
\end{aligned}
\right.\tag{23}
\end{equation}    
\\
Therefore, the equilibrium point $\mathcal{E}_4$  \eqref{eq23} exists and is biologically feasible if $\beta q b > (a_1 \gamma_1 \gamma_2 k + P_1 b) v_1$ and $v_1>0.$
\subsubsection{Local  Stability of Equilibrium Points}

In this section, we investigate the local stability of the four equilibrium points of system \eqref{equation7} using the Jacobian matrix \eqref{equation24}  . The system is first linearized about each equilibrium point and the corresponding eigenvalues are then used to determine whether the equilibrium point is locally asymptotically stable or unstable. An equilibrium point is locally asymptotically stable if all eigenvalues of the Jacobian matrix evaluated at that point have negative real parts, otherwise, the equilibrium point is unstable \cite{Hunaish2025}. The Jacobian matrix of system \eqref{equation7} is given by

{\scriptsize
\[
J =
\left|
\begin{array}{cccccccc}\label{equation24}
b -\frac{ 2bM_1}{k} - (\gamma_1L_1 + \gamma_2L_2) & 0 & -\gamma_1M_1 & -\gamma_2M_1 & 0 \\
0 & -(\beta + \mu_1) & 0 & 0 & \alpha \lambda \\
a_1\gamma_1L_1 & \beta & (a_1\gamma_1M_1 - \delta_1 - P_1L_2 - \mu_2) & -P_1L_1 & 0 \\
a_2\gamma_2L_2 & 0 & (\delta_1 + P_1L_2) & (a_2\gamma_2M_1 + P_1L_1 - 2P_2L_2 - \delta_2 - \mu_3) & 0 \\
0 & 0 & 0 & \delta_2 & -\mu_4 \\
\end{array}
\right|.\tag{24}
\]
}
\begin{thm}\label{theorem3}
    The trivial equilibrium point $\mathcal{E}_1$ is locally unstable whenever $b>0$.
\end{thm}
\begin{proof}
        The jacobian of system \eqref{equation24} associated with the equilibrium $\mathcal{E}_1$  is
\[\label{equation25}
J(\mathcal{E}_1) = 
\left|
\begin{array}{cccccc}
b & 0 & 0 & 0 & 0 \\
0 & -(\beta + \mu_1) & 0 & 0 & \alpha \lambda \\
0 & \beta & -(\delta_1 + \mu_2) & 0 & 0 \\
0 & 0 & \delta_1 & -(\delta_2 + \mu_3) & 0 \\
0 & 0 & 0 & \delta_2 & -\mu_4
\end{array}\tag{25}
\right|.
\]
It is immediate that one of the eigenvalues of \eqref{equation25} is $b>0$. The remaining four eigenvalues are obtained from the characteristic equation 
\begin{equation}\label{equation26}
\begin{aligned}
x^4 + c_1 x^3 + c_2 x^2 + c_3 x + c_4 = 0,
\end{aligned}\tag{26}
\end{equation}
with
\[
\left\{
\begin{aligned}\label{eqtn27}
c_1 &= m_1 + m_2 + m_3 + m_4, \\
c_2 &= m_1 m_2 + m_3 m_4 + (m_1 + m_2)(m_3 + m_4), \\
c_3 &= m_1 m_2 (m_3 + m_4) + m_3 m_4 (m_1 + m_2) - \alpha \lambda \beta \delta_1 \delta_2, \\
c_4 &= m_1 m_2 m_3 m_4,
\end{aligned}\tag{27}
\right.
\]
where 
$m_1 = \beta + \mu_1$,
$m_2 = \delta_1 + \mu_2,$
$m_3 = \delta_2 + \mu_3$ and 
$m_4 = \mu_4.$
Applying the Routh–Hurwitz criteria as presented in \cite{Allen2007}, the equilibrium point $\mathcal{E}_1$ is locally asymptotically stable if 
\begin{equation}\label{equation28}
\left\{
\begin{aligned}
H_1 &:\quad c_1>0,\; c_2>0,\; c_3>0,\; c_4>0,\\
H_2 &:\quad c_1c_2-c_3>0,\\
H_3 &:\quad c_1c_2c_3-c_1^2c_4-c_3^2>0.
\end{aligned}
\right.
\tag{28}
\end{equation}
It is evident that $c_1$, $c_2$ and $c_4$ are positive, while the positivity of $c_3$ depends on the model parameters. However, since one eigenvalue of the Jacobian matrix \eqref{equation24} is b, it is always positive. Therefore, we do not investigate the remaining Routh–Hurwitz conditions and hence we conclude that the trivial equilibrium point $\mathcal{E}_1$ is locally unstable. This completes the proof.
\end{proof}
\begin{thm}
     Suppose that $b>0$. If the coefficients $y_1$, $y_2$, $y_3$ and $y_4$ defined in \eqref{equation31} satisfy the Routh--Hurwitz conditions in \eqref{equation32}, then the equilibrium point $\mathcal{E}_2$ is locally asymptotically stable; otherwise, it is unstable.
\end{thm}
\begin{proof}
    Corresponding to the equilibrium point $\mathcal{E}_2$, the Jacobian of system \eqref{equation24} is
\[
J (\mathcal{E}_2)=
\left|
\begin{array}{cccccccc}\label{equation29}
-b & 0 & -\gamma_1 k & -\gamma_2 k & 0 \\
0 & -(\beta + \mu_1) & 0 & 0 & \alpha \lambda \\
0 & \beta & (a_1 \gamma_1 k - \delta_1 - \mu_2) & 0 & 0 \\
0 & 0 & \delta_1 & (a_2 \gamma_2 k - \delta_2 - \mu_3) & 0 \\
0 & 0 & 0 & \delta_2 & -\mu_4 \\
\end{array}
\right|.\tag{29}
\]
We can observe that one of the eigenvalues of  \eqref{equation29} is  $-b < 0$. The other eigenvalues are obtained from the characteristic equation
\begin{equation}
    \begin{aligned}\label{equation30}
        x^4 + y_1 x^3 + y_2 x^2 + y_3 x + y_4 = 0,
    \end{aligned}\tag{30}
\end{equation}
with
\begin{equation}
\left\{
\begin{aligned}\label{equation31}
y_1 &= m_1 + m_2 + m_3 + m_4, \\
y_2 &= m_1m_2+(m_1 + m_2)(m_3 + m_4) + m_3 m_4, \\
y_3 &= m_1\left[m_2(m_3 + m_4) + m_3 m_4\right] + m_2 m_3 m_4 - \alpha \lambda \beta \delta_1 \delta_2, \\
y_4 &= m_1 m_2 m_3 m_4 - \alpha \lambda \beta \delta_1 \delta_2,
\end{aligned}\tag{31}
\right.
\end{equation}
where
$m_1 = \beta + \mu_1,$ 
$m_2 = \delta_1 + \mu_2 - a_1 \gamma_1 k,$
$m_3 = \delta_2 + \mu_3 - a_2 \gamma_2 k$ and 
$m_4 = \mu_4.$
The Routh–Hurwitz criteria for the local asymptotic stability of $\mathcal{E}_2$ require that 
\begin{equation}\label{equation32}
\left\{
\begin{aligned}
H_1 &:\quad y_1>0,\; y_2>0,\; y_3>0,\; y_4>0,\\
H_2 &:\quad y_1y_2-y_3>0,\\
H_3 &:\quad y_1y_2y_3-y_1^2y_4-y_3^2>0.
\end{aligned}
\right.
\tag{32}
\end{equation}
If all conditions in \eqref{equation32} are satisfied, then all the eigenvalues of the Jacobian matrix have negative real parts. Hence, the equilibrium point $\mathcal{E}_2$ is locally asymptotically stable. This completes the proof.
\end{proof}
\begin{thm}
     The maize extinction equilibrium point $\mathcal{E}_3$ is locally asymptotically stable whenever $b <(\gamma_1 L_1^3 + \gamma_2 L_2^3)$ and conditions in \eqref{equation36} hold; otherwise, it is unstable.
\end{thm}
\begin{proof}
    Evaluating the Jacobian matrix \eqref{equation24} about the equilibrium point $\mathcal{E}_3$ yields
{\scriptsize
\begin{align*}\label{equation33}
J(\mathcal{E}_3) =
\begin{bmatrix}
b - (\gamma_1 L_1^3 + \gamma_2 L_2^3) & 0 & 0 & 0 & 0 \\
0 & -(\beta + \mu_1) & 0 & 0 & \alpha\lambda \\
a_1 \gamma_1 L_1^3 & \beta & -(\delta_1 + P_1 L_2^{3} + \mu_2) & -P_1 L_1^3 & 0 \\
a_2 \gamma_2 L_2^3 & 0 & \delta_1 + P_1 L_2^{3} & (P_1 L_1^3 - 2 P_2 L_2^{3} - \delta_2 - \mu_3) & 0 \\
0 & 0 & 0 & \delta_2 & -\mu_4 \\
\end{bmatrix}.\tag{33}
\end{align*}
}  
It is immediate that one of the eigenvalue of \eqref{equation33} is  $b - (\gamma_1 L_1^3 + \gamma_2 L_2^3)$. The other eigenvalues are obtained from the characteristic equation
\begin{equation}\label{eq34}
\begin{aligned}
x^{4} + q_1 x^{3}+ q_2 x^{2} + q_3 x + q_4 = 0,    
\end{aligned}\tag{34}
\end{equation}
with
\[
\left\{
\begin{aligned}\label{35}
q_1 &= m_1 + m_2 + m_3 + m_4, \\
q_2 &= m_1 m_2 + m_1 m_3 + m_1 m_4 
     + m_2 m_3 + m_2 m_4 + m_3 m_4+n_3n_5, \\
q_3 &= m_1 m_2 m_3 + m_1 m_2 m_4 
     + m_1 m_3 m_4 + m_2 m_3 m_4+n_3n_5(m_1+m_4), \\
q_4 &= m_1 m_2 m_3 m_4+m_1m_4n_3n_5 - n_1 n_2 n_3 n_4,
\end{aligned}\tag{35}
\right.
\]
where
$m_1 = \beta + \mu_1$,
$m_2 = \delta_1 + P_1 L_2^3 + \mu_2,$
$m_3 = \delta_2 + \mu_3 + 2P_2 L_2^3 - P_1 L_1^3,$
$m_4 = \mu_4,$
$n_1 = \beta,$
$n_2 = \alpha\lambda,$
$n_3 = \delta_1 + P_1 L_2^3,$
$n_4 = \delta_2$ and
$n_5=P_1L^3_1$.
The Routh–Hurwitz criteria for the local asymptotic stability of $\mathcal{E}_3$ require that 
\begin{equation}\label{equation36}
\left\{
\begin{aligned}
H_1 &:\quad q_1>0,\; q_2>0,\; q_3>0,\; q_4>0,\\
H_2 &:\quad q_1q_2-q_3>0,\\
H_3 &:\quad q_1q_2q_3-q_1^2q_4-q_3^2>0.
\end{aligned}
\right.
\tag{36}
\end{equation}
 The eigenvalue $b -(\gamma_1 L_1^3 + \gamma_2 L_2^3)$ is negative whenever $b < (\gamma_1 L_1^3 + \gamma_2 L_2^3)$. Furthermore, if the conditions in \eqref{equation36} hold, then all the roots in \eqref{eq34} have negative real parts. Hence, all eigenvalues of the jacobian matrix \eqref{equation33} have negative real parts and therefore $\mathcal{E}_3$ is locally asymptotically stable. This completes the proof.
\end{proof}   
\begin{thm}\label{theorem6}
    The maize co-existence equilibrium point $\mathcal{E}_4$ is locally asymptotically stable if conditions \eqref{equation40} are satisfied; otherwise, it is unstable.
\end{thm}
\begin{proof}
    Evaluating the Jacobian matrix \eqref{equation24} about the equilibrium point $\mathcal{E}_4$ yields
{\scriptsize
\begin{align*}\label{equation37}
J(\mathcal{E}_4) =
\begin{bmatrix}
b-\frac{ 2bM_1^4}{k} - (\gamma_1 L_1^4 + \gamma_2 L_2^4) &0& -\gamma_1 M_1^4 & -\gamma_2 M_1^4 & 0 \\
0 & -(\beta + \mu_1) & 0 & 0 & \alpha\lambda \\
a_1 \gamma_1 L_1^4 & \beta & (a_1 \gamma_1 M_1^4 - \delta_1 - P_1 L_2^4 - \mu_2) & -P_1 L_1^4 & 0 \\
a_2 \gamma_2 L_2^4 & 0 & (\delta_1 + P_1 L_2^4) & (a_2 \gamma_2 M_1^4 + P_1 L_1^4 - 2P_2 L_2^4 -\delta_2- \mu_3) & 0 \\
0 & 0 & 0 & \delta_2 & -\mu_4
\end{bmatrix}.\tag{37}
\end{align*}
}
The characteristic equation of the jacobian matrix \eqref{equation37} is
\begin{equation}
\begin{aligned}\label{equation38}
x^{5} + r_1 x^{4} + r_2 x^{3} + r_3 x^{2} + r_4 x + r_5 = 0   
\end{aligned}\tag{38}
\end{equation}
with
\[
\left\{
\begin{aligned}\label{equation39}
r_1 &= m_1 + m_2 + m_3 + m_4 - n_9, \\
r_2 &= (m_2 m_3 + n_3 n_4) + (m_1 + m_4)(m_2 + m_3) + m_1 m_4 - n_9(m_1 + m_2 + m_3 + m_4), \\
r_3 &= (m_1 + m_4)(m_2 m_3 + n_3 n_4) + m_1 m_4(m_2 + m_3) - n_9\big[(m_2 m_3 + n_3 n_4) + (m_1 + m_4)(m_2 + m_3) + m_1 m_4\big]\\
&\quad + n_1 n_6 \delta_2, \\
r_4 &= m_1 m_4(m_2 m_3 + n_3 n_4) - n_7 n_8 n_4 \delta_2 - n_9\big[(m_1 + m_4)(m_2 m_3 + n_3 n_4) + m_1 m_4(m_2 + m_3)\big] \\
&\quad+ n_1 \delta_2(m_1 n_6 + n_5 n_4 + n_6 m_2), \\
r_5 &= n_1 \delta_2(m_1 n_5 n_4 + m_1 n_6 m_2) - n_9\big[m_1 m_4(m_2 m_3 + n_3 n_4) - n_7 n_8 n_4 \delta_2\big],
\end{aligned}\tag{39}
\right.
\]
where
$m_1 = \beta + \mu_1,$ 
$m_2 = \delta_1 + P_1 L_2^4 + \mu_2 - a_1 \gamma_1 M_1^4, $
$m_3 = \delta_2 + \mu_3 + 2P_2 L_2^4 - P_1 L_1^4 - a_2 \gamma_2 M_1^4,$
$m_4 = \mu_4, $
$n_1 = \gamma_1 M_1^4, $
$n_2 = \gamma_2 M_1^4,$
$n_3 = P_1 L_1^4,$
$n_4 = \delta_1 + P_1 L_2^4,$
$n_5 = a_1 \gamma_1 L_1^4,$
$n_6 = a_2 \gamma_2 L_2^4,$
$n_7 = \alpha \lambda,$
$n_8 = \beta$ and
$n_9 = b - \frac{2bM_1^4}{k} - (\gamma_1 L_1^4 + \gamma_2 L_2^4).$
According to the Routh--Hurwitz criteria, the necessary and sufficient conditions for the local asymptotic stability of the equilibrium point $\mathcal{E}_4$ are that the Hurwitz determinants associated with the characteristic polynomial \eqref{equation38} are positive \cite{Lancaster1969}. For a fifth-degree polynomial, these conditions are
\begin{equation}\label{equation40}
\left\{
\begin{aligned}
H_1 &:\quad r_1>0,\; \\
H_2 &:\quad r_1r_2-r_3>0,\\
H_3 &:\quad r_1r_2r_3-r_1^{2}r_4-r_3^{2}+r_1r_5>0,\\
H_4 &:\quad (r_3r_4-r_2r_5)(r_1r_2-r_3)-(r_1r_4-r_5)^2>0,\\
H_5 &:\quad r_5H_4>0.
\end{aligned}
\right.
\tag{40}
\end{equation}
Therefore, the equilibrium point $\mathcal{E}_4$ is locally asymptotically stable provided that conditions in \eqref{equation40} are satisfied; otherwise, it is unstable. Hence, the theorem holds.
\end{proof}
\subsubsection{Global Stability of Equilibrium Points}\label{section2.1.6}
In this section, we construct Lyapunov functions to analyze the global stability of the equilibrium points obtained in subsubsection \eqref{section2.1.4}.
\\
\\
\textbf{(a)} The trivial equilibrium point $\mathcal{E}_1$. Since $\mathcal{E}_1$ is locally unstable as established by the positive eigenvalue $b > 0$, it cannot be globally asymptotically stable. Therefore, no further investigation of its global stability is undertaken.
\\
\\
\textbf{(b)} The non trivial equilibrium point $\mathcal{E}_2$.  We Consider the Lyapunov function
\[ \tag{41}\label{equation41}
U_2(M_1, E, L_1, L_2, A) = M_1 - k - k \ln\left(\frac{M_1}{k}\right) + \frac{m_4}{\alpha \lambda} E + \frac{m_1 m_4}{\alpha \lambda \beta} L_1 + \frac{m_1 m_4}{\alpha \lambda \delta_1} L_2 + \frac{m_1 m_3 m_4}{\alpha \lambda \delta_1 \delta_2} A,
\]
where $m_1 = \beta + \mu_1,$ $m_2 = \delta_1 + \mu_2$, $m_3 = \delta_2 + \mu_3$ and $m_4 = \mu_4.$
It is evident that the function $U_2(M_1, E, L_1, L_2, A) $ is positive definite, continuous and vanishes at $\mathcal{E}_2$. Taking the derivative along the solutions of system \eqref{equation7} yields
\[ \tag{42}\label{equation42}
\begin{aligned}
\frac{dU_2}{dt}
&\leq
\left(1-\frac{k}{M_1}\right)\frac{dM_1}{dt}
+\frac{m_4}{\alpha\lambda}\frac{dE}{dt}
+\frac{m_1m_4}{\alpha\lambda\beta}\frac{dL_1}{dt} 
+\frac{m_1m_4}{\alpha\lambda\delta_1}\frac{dL_2}{dt}
+\frac{m_1m_3m_4}{\alpha\lambda\delta_1\delta_2}\frac{dA}{dt}.
\end{aligned}
\]
\[ 
\begin{aligned}
\frac{dU_2}{dt}
&=
-\frac{b}{k}(M_1-k)^2 
+ \gamma_1 L_1 \left[ k + M_1 \left( \frac{m_1 m_4 a_1}{\alpha\lambda\beta} - 1 \right) \right] 
+ \gamma_2 L_2 \left[ k + M_1 \left( \frac{m_1 m_4 a_2}{\alpha\lambda\delta_1} - 1 \right) \right] \\
&\quad
+ m_4 A \left( 1 - \frac{m_1 m_3 m_4}{\alpha\lambda\delta_1\delta_2} \right) 
- \frac{2m_1 m_4 P_2}{\alpha\lambda\delta_1} L_2^2 
+ \frac{m_1 m_4}{\alpha\lambda} L_1 \left( 1 - \frac{m_2}{\beta} \right) 
+ \frac{m_1 m_4 P_1}{\alpha\lambda} L_1 L_2 \left( \frac{1}{\delta_1} - \frac{1}{\beta} \right).
\end{aligned}
\]
The non-trivial equilibrium point $\mathcal{E}_2$ is globally asymptotically stable if the following conditions hold:
$M_1(t) \geq k$,
$m_1 m_3 m_4 \leq \alpha\lambda\delta_1\delta_2$,
$m_1 m_4 a_1 \leq \alpha\lambda\beta$,
$m_1 m_4 a_2 \leq \alpha\lambda\delta_1$ and $\beta \leq \delta_1$; otherwise, it is unstable.
Under these conditions, $\frac{dU_2}{dt} \leq 0$.
\\
\\
\textbf{(c)}To investigate the possibility of establishing the global asymptotic stability of the maize extinction equilibrium point $\mathcal{E}_3$ and the co-existence equilibrium point $\mathcal{E}_4$, we consider the generalized Lyapunov function
\[ \tag{43}\label{equation43}
\begin{aligned}
U_3(M_1, E, L_1, L_2, A) &= f_0\left(M_1 - M_1^* - M_1^*\ln\frac{M_1}{M_1^*}\right) + f_1\left(E - E^* - E^*\ln\frac{E}{E^*}\right) + f_2\left(L_1 - L_1^* - L_1^*\ln\frac{L_1}{L_1^*}\right) \\
&\quad + f_3\left(L_2 - L_2^* - L_2^*\ln\frac{L_2}{L_2^*}\right)  + f_4\left(A - A^* - A^*\ln\frac{A}{A^*}\right),
\end{aligned}
\]
where $f_0$, $f_1$, $f_2$, $f_3$, $f_4$ are arbitrary positive constants. Taking the derivative  along the solutions of system \eqref{equation7} yields
\[
\begin{aligned} \label{equation44}
\frac{dU_3}{dt} &\leq f_0\left(1-\frac{M_1^*}{M_1}\right)\frac{dM_1}{dt}
+ f_1\left(1-\frac{E^*}{E}\right)\frac{dE}{dt}
+ f_2\left(1-\frac{L_1^*}{L_1}\right)\frac{dL_1}{dt}  
+ f_3\left(1-\frac{L_2^*}{L_2}\right)\frac{dL_2}{dt}
+ f_4\left(1-\frac{A^*}{A}\right)\frac{dA}{dt}.
\end{aligned} \tag{44}\
\]
\[ \tag{45}\label{equation45}
\begin{aligned}
\frac{dU_3}{dt} &= f_0 b (M_1 - M_1^*)\left(1 - \frac{M_1}{k}\right) - f_0 \gamma_1 (M_1 - M_1^*) L_1 - f_0 \gamma_2 (M_1 - M_1^*) L_2 - f_1 m_1 (E - E^*) + f_1 \alpha \lambda \left(1 - \frac{E^*}{E}\right) A \\
&\quad + f_2 \beta (E - E^*) + f_2 a_1 \gamma_1 (M_1 L_1 - M_1 L_1^*) - f_2 m_2 (L_1 - L_1^*) - f_2 P_1 (L_1 L_2 - L_1^* L_2)  + f_3 \delta_1 (L_1 - L_1^*) \\
&\quad + f_3 P_1 (L_1 L_2 - L_1 L_2^*) + f_3 a_2 \gamma_2 (M_1 L_2 - M_1 L_2^*) - f_3 m_3 (L_2 - L_2^*) - 2 f_3 P_2 (L_2^2 - L_2 L_2^*)  + f_4 \delta_2 (L_2 - L_2^*)\\
&\quad -f_4 m_4 (A - A^*).
\end{aligned}
\]
From equation \eqref{equation45}, it is evident that $\frac{dU_3}{dt}$ contains nonlinear interaction terms, 
whose signs cannot, in general be determined . Consequently, it is not possible to establish that $\frac{dU_3}{dt}\leq0$ throughout the feasible region. Therefore, the proposed Lyapunov function is insufficient to establish the global asymptotic stability of $\mathcal{E}_3$ and $\mathcal{E}_4$. Hence, the global stability of these equilibrium points remains unresolved for the present model.

\section{The Optimal Control Problem}\label{sect3}
In this section, an optimal control framework is introduced to determine optimal strategies for minimizing the FAW population while accounting for the costs associated with implementing control measures. The construction of an optimal control problem is necessary because continuous and excessive use of control measures, especially chemical pesticides, can be costly and environmentally harmful. To achieve this, the model is reformulated for both the vegetative and reproductive stages by incorporating two time-dependent control parameters, $u_1(t)$ and $u_2(t)$. Control $u_1(t)$ represents the effort associated with traditional control methods, such as handpicking and destruction of FAW egg masses and early instar larvae, while control $u_2(t)$ represents the application of chemical pesticides targeting late instar larvae. Incorporating these controls into the deterministic model yields the following controlled system of differential equations. During the vegetative stage, the controlled system is given by

\begin{equation}\label{equation46}
\left\{
\begin{aligned}
\frac{dM_1}{dt} &= b M_1 \left(1 - \frac{M_1}{k}\right) - (\gamma_1 L_1 + \gamma_2 L_2) M_1 ,\\
\frac{dE}{dt} &= \alpha \lambda A - (\beta + \mu_1 + u_1)E, \\
\frac{dL_1}{dt} &= \beta E + a_1\gamma_1 L_1 M_1 - P_1 L_1L_2 - (\delta_1 + \mu_2 + u_1) L_1, \\
\frac{dL_2}{dt} &= a_2 \gamma_2 L_2 M_1 + P_1 L_1L_2 + \delta_1 L_1 - P_2 L^2_2 - (\delta_2 + \mu_3 + u_2)L_2 ,\\
\frac{dA}{dt} &= \delta_2 L_2 - \mu_4 A.
\end{aligned}
\right.
\tag{46}
\end{equation}
To ensure biological feasibility, the initial conditions satisfy
\[
M_1(0)=k,\quad \quad E(0) \geq 0,\quad L_1(0) \geq 0,\quad L_2(0) \geq 0,\quad A(0) \geq 0.
\]
Similarly, during the reproductive stage, the dynamics are governed by
\begin{equation}\label{equation47}
\left\{
\begin{aligned}
\frac{dM_2}{dt} &= b M_2 \left(1 - \frac{M_2}{k}\right) - (\gamma_3 L_1 + \gamma_4 L_2) M_2 ,\\
\frac{dE}{dt} &= \alpha \lambda A - (\beta + \mu_1 + u_1)E, \\
\frac{dL_1}{dt} &= \beta E + a_3\gamma_3 L_1 M_2 - P_1 L_1L_2 - (\delta_1 + \mu_2 + u_1) L_1, \\
\frac{dL_2}{dt} &= a_4 \gamma_4 L_2 M_2 + P_1 L_1L_2 + \delta_1 L_1 - P_2 L^2_2 - (\delta_2 + \mu_3 + u_2)L_2 ,\\
\frac{dA}{dt} &= \delta_2 L_2 - \mu_4 A,
\end{aligned}
\right.
\tag{47}
\end{equation}
with the transition conditions
$M_1(t_1)=M_2(t_1)$ and $M_2(t)=0$ for 
$t<t_1$, ensuring continuity between the two maize growth stages. The admissible control sets for the vegetative and reproductive stages are defined, respectively. The admissible control set is defined by
\[\label{equation48}
\Gamma = \left\{ (u_1(t), u_2(t)) \,\middle|\, 0 \leq u_1(t) \leq u_{1,\text{max}},\; 0 \leq u_2(t) \leq u_{2,\text{max}} \right\},\tag{48}
\]
where $u_{1,\text{max}}$ and $u_{2,\text{max}}$ denote the maximum allowable levels of the traditional and chemical control strategies, respectively. For the vegetative stage, $u_{1,\text{max}}=0.4$ and $u_{2,\text{max}}=0.5$, while for the reproductive stage, $u_{1,\text{max}}=0.1$ and $u_{2,\text{max}}=0.2$.
The primary objective is to suppress the FAW population at its early and most destructive stages while limiting the cost of implementing the control measures. Accordingly, the objective function is defined as
\[\label{equation49}
J(u_1, u_2) = \int_0^T \left( C_1 E(t) + C_2 L_1(t) + C_3 L_2(t) + \frac{W_1}{2} u_1^2(t) + \frac{W_2}{2} u_2^2(t) \right) dt,\tag{49}
\]
subject to the system of differential equations in systems \eqref{equation46} and \eqref{equation47}. $C_1$, $C_2$ and $C_3$ are non-negative weights associated with the egg, early instar larvae and late instar larvae, respectively, while $W_1$ and $W_2$ are positive weights associated with the control efforts. The quadratic term ,$\frac{W_1}{2} u_1^2(t) + \frac{W_2}{2} u_2^2(t)$ regulate the intensity of the control effort while ensuring that the controls remain bounded. Its quadratic structure also preserves the smoothness and convexity of the Hamiltonian which guarantees the uniqueness of the optimal control. The main objective is therefore to determine an optimal pair of control $(u_1^*, u_2^*)$ that effectively minimizes the population of FAW eggs and larvae. By Pontryagin's Maximum Principle \cite{Bettiol2021}, the Hamiltonian associated with system \eqref{equation46} is defined as
\begin{equation}\label{equation50}
\left\{
\begin{aligned}
H &= C_1 E + C_2 L_1 + C_3 L_2 + \frac{W_1}{2} u_1^2 + \frac{W_2}{2} u_2^2 \\
& +  \lambda_1[  b M_1 \left(1 - \frac{M_1}{k}\right) - (\gamma_1 L_1 + \gamma_2 L_2) M_1 ] \\
&+\lambda_2[ \alpha \lambda A - (\beta + \mu_1+ u_1) E]\\
&+\lambda_3[\beta E + a_1\gamma_1 L_1 M_1 - P_1 L_1L_2 - (\delta_1 + \mu_2 + u_1) L_1 ]\\
&+\lambda_4[a_2 \gamma_2 L_2 M_1 + P_1 L_1L_2 + \delta_1 L_1 - P_2 L^2_2 - (\delta_2 + \mu_3 + u_2)L_2]\\
&+\lambda_5[ \delta_2 L_2 - \mu_4 A].
\end{aligned}
\right.
\tag{50}
\end{equation}
The adjoint variables $\lambda_i(t)$, corresponding to the state variables $x_i = M_1, E, L_1, L_2, A$, satisfy the adjoint system 
\[
\frac{d\lambda_i}{dt} = -\frac{\partial H}{\partial x_i} \qquad i=1,2,3,4,5.
\]
Applying Pontryagin's Maximum Principle to the Hamiltonian \eqref{equation50} yields the following adjoint equations
\begin{equation}\label{equation51}
\left\{
\begin{aligned}
\frac{d\lambda_1}{dt} = \frac{\partial H}{\partial M_1} &= -\left[ \lambda_1 \left( b \left(1 - \frac{2M_1}{k}\right) - (\gamma_1 L_1 + \gamma_2 L_2) \right) + \lambda_3 a_1 \gamma_1 L_1 + \lambda_4 a_2 \gamma_2 L_2\right], \\
\frac{d\lambda_2}{dt} = \frac{\partial H}{\partial E} &= -\left[ C_1 - \lambda_2 (\beta + \mu_1 + u_1)  + \lambda_3 \beta \right], \\
\frac{d\lambda_3}{dt} = \frac{\partial H}{\partial L_1} &= -\left[ C_2 - \lambda_1 \gamma_1 M_1 + \lambda_3 (a_1 \gamma_1 M_1 - P_1L_2 - (\delta_1 + \mu_2 + u_1))  + \lambda_4 ( P_1L_2 + \delta_1) \right],\\
\frac{d\lambda_4}{dt} = \frac{\partial H}{\partial L_2} &= -\left[ C_3 - \lambda_1 \gamma_2 M_1 - \lambda_3 P_1L_1 + \lambda_4 (a_2 \gamma_2 M_1 + P_1L_1 -2P_2L_2-(\delta_2+\mu_3+u_2)) +\lambda_5 \delta_2 \right], \\
\frac{d\lambda_5}{dt} =  \frac{\partial H}{\partial A} &= -\left[\lambda_2 \alpha\lambda - \lambda_5 \mu_4\right],
\end{aligned}
\right.
\tag{51}
\end{equation}
subject to the transversality condition $\lambda_i(T)=0$ for  $i= 1,2,3,4,5$, where \(T\) denotes the final simulation time. The optimal controls $u^*_1$ and $u^*_2$ are then obtained as
\[
\frac{\partial H}{\partial u_1}=0
\quad \Rightarrow \quad
u_1^*(t)=\frac{E(t)\lambda_2(t) + L_1(t)\lambda_3(t)}{W_1} \qquad \text{and} \qquad
\frac{\partial H}{\partial u_2}=0
\quad \Rightarrow \quad
u_2^*(t)=\frac{L_2(t)\lambda_4(t)}{W_2}.
\]
Since the controls are restricted to the admissible set $\Gamma$ defined in \eqref{equation48}, the optimal controls are projected onto their respective lower and upper bounds. Thus, the bounded optimal controls are characterized by
\begin{equation}\label{equation52}
\left\{
\begin{aligned}
u_1^*(t) &= \min\left\{ u_{1,\max},
\max\left( \frac{E(t)\lambda_2(t)+L_1(t)\lambda_3(t)}{W_1},0\right)\right\},\\[1ex]
u_2^*(t) &= \min\left\{ u_{2,\max},
\max\left( \frac{L_2(t)\lambda_4(t)}{W_2},0\right)\right\}.
\end{aligned}
\tag{52}
\right.
\end{equation}
The expressions in \eqref{equation52} ensure that the controls remain within their biologically and economically feasible bounds throughout the simulation period.
Using the same approach, the Hamiltonian associated with system \eqref{equation47} takes the form
\begin{equation}\label{equation53}
\left\{
\begin{aligned}
H &= C_1 E + C_2 L_1 + C_3 L_2 + \frac{W_1}{2} u_1^2 + \frac{W_2}{2} u_2^2 \\
& +  \lambda_1[   b M_2 \left(1 - \frac{M_2}{k}\right) - (\gamma_3 L_1 + \gamma_4 L_2) M_2 ] \\
&+\lambda_2[ \alpha \lambda A - (\beta + \mu_1+u_1) E]\\
&+\lambda_3[ \beta E + a_3\gamma_3 L_1 M_2 - P_1 L_1L_2 - (\delta_1 + \mu_2 + u_1) L_1 ]\\
&+\lambda_4[a_4 \gamma_4 L_2 M_2 + P_1 L_1L_2 + \delta_1 L_1 - P_2 L^2_2 - (\delta_2 + \mu_3 + u_2)L_2]\\
&+\lambda_5[\delta_2 L_2 - \mu_4 A].
\end{aligned}
\right.
\tag{53}
\end{equation}
The adjoint variables $\lambda_i(t)$, associated with the state variables $M_2$, $E$, $L_1$, $L_2$ and $A$, satisfy the following adjoint equations:
\begin{equation}\label{equation54}
\left\{
\begin{aligned}
\frac{d\lambda_1}{dt} = \frac{\partial H}{\partial M_2} &= -\left[ \lambda_1 \left( b \left(1 - \frac{2M_2}{k}\right) - (\gamma_3 L_1 + \gamma_4 L_2) \right) + \lambda_3 a_3 \gamma_3 L_1 + \lambda_4 a_4 \gamma_4 L_2\right], \\
\frac{d\lambda_2}{dt} = \frac{\partial H}{\partial E} &= -\left[ C_1 - \lambda_2 (\beta + \mu_1 + u_1)  + \lambda_3 \beta \right], \\
\frac{d\lambda_3}{dt} = \frac{\partial H}{\partial L_1} &= -\left[ C_2 - \lambda_1 \gamma_3 M_2 + \lambda_3 (a_3 \gamma_3 M_2 - P_1L_2 - (\delta_1 + \mu_2 + u_1))  + \lambda_4 ( P_1L_2 + \delta_1) \right],\\
\frac{d\lambda_4}{dt} = \frac{\partial H}{\partial L_2} &= -\left[ C_3 - \lambda_1 \gamma_4 M_2 - \lambda_3 P_1L_1 + \lambda_4 (a_4 \gamma_4 M_2 + P_1L_1 -2P_2L_2-(\delta_2+\mu_3+u_2)) +\lambda_5 \delta_2 \right], \\
\frac{d\lambda_5}{dt} =  \frac{\partial H}{\partial A} &= -\left[\lambda_2 \alpha\lambda - \lambda_5 \mu_4\right],
\end{aligned}
\right.
\tag{54}
\end{equation}
subject to the transversality condition $\lambda_i(T)=0$ for $i= 1,2,3,4,5$, where \(T\) denotes the final simulation time. The optimal controls $u^*_1$ and $u^*_2$ are then determined as
\[
\frac{\partial H}{\partial u_1}=0
\quad \Rightarrow \quad
u_1^*(t)=\frac{E(t)\lambda_2(t) + L_1(t)\lambda_3(t)}{W_1} \qquad \text{and} \qquad
\frac{\partial H}{\partial u_2}=0
\quad \Rightarrow \quad
u_2^*(t)=\frac{L_2(t)\lambda_4(t)}{W_2}.
\]
Furthermore, the bounded optimal controls are characterized by
\[
u_1^*(t)=\min\left\{u_{1,\max},\max\left(\frac{E(t)\lambda_2(t)+L_1(t)\lambda_3(t)}{W_1},0\right)\right\},\qquad
u_2^*(t)=\min\left\{u_{2,\max},\max\left(\frac{L_2(t)\lambda_4(t)}{W_2},0\right)\right\}.
\]
\section{Numerical Simulations}\label{sect4}
\subsection{Model Parametrization}\label{subsect4.1}
This section provides an overview of the baseline parameter values used in the model. Some of the parameter values were adopted from the existing literature. For parameters with no experimentally reported or published values, biologically plausible values were assumed based on the known biology and life cycle of the Fall Armyworm. These assumptions provide a representative baseline for the numerical simulations and improve the reproducibility of the model results. A brief justification of the parameter values is presented below.
\begin{enumerate}
    \item \textbf{Natural mortality rates ($\mu_1,\mu_2,\mu_3,\mu_4$).} Following Daudi et al.~\cite{daudi_modelling_2021}, the natural mortality rate is defined as the reciprocal of the expected lifespan of an organism. Hence,
    \[\tag{55}\label{55}
    \text{Mortality rate} = \frac{1}{\text{expected lifetime}}.
    \]
    The mortality rates employed in this study, including both adopted and assumed values, are presented in Table~\ref{table4}.
    \item \textbf{Egg-laying rate ($\lambda$).} 
    The egg-laying rate was adopted from Daudi et al. \cite{daudi_modelling_2021} and incorporated into the model as the effective daily oviposition rate of adult moths. Thus, $\lambda$ represents the rate of egg production used in the mathematical formulation and is not intended to represent the total lifetime fecundity of an adult female moth. The adopted value is provided in Table \ref{table4}.
    \item \textbf{Egg hatching and larval development rates ($\beta,\delta_1,\delta_2$).} The transition rates $\beta$, $\delta_1$, and $\delta_2$ were adopted directly from Daudi et al.~\cite{daudi_modelling_2021} and are used as effective stage-transition parameters within the mathematical formulation. In particular, $\beta$ is treated as a model-specific egg-to-larva transition rate and is not interpreted as the direct reciprocal of the observed biological egg-stage duration. The biological durations reported in the literature are therefore provided as reference values for the life-history stages rather than as direct numerical equivalents of the transition rates adopted in the model. The egg stage typically lasts 2--4 days depending on climatic conditions \cite{assefa2020status}, while the larval stage lasts about 14 days under warm conditions and up to 30 days during cooler weather \cite{adjaoke2022fall}.
    \item \textbf{Maize destruction rates ($\gamma_1,\gamma_2,\gamma_3,\gamma_4$).} These parameters were assumed due to lack of reliable experimental estimates and were calibrated within biologically plausible ranges to produce biologically realistic maize damage and FAW population dynamics. The selected values reflect the greater feeding capacity of late instar larvae compared to early instar larvae during both the vegetative and reproductive stages. 
    \item \textbf{Feed conversion rates ($a_1,a_2,a_3,a_4$).} These parameters were calibrated within biologically plausible ranges to produce realistic maize consumption and FAW population dynamics. Higher conversion rates were assigned to late instar larvae to reflect their greater feeding efficiency. 
    \item \textbf{Cannibalism rates ($P_1,P_2$).} These parameters were assumed based on the documented cannibalistic behaviour of larvae. Their values were selected within biologically plausible ranges and calibrated through preliminary simulations to reproduce realistic interactions between the early instar larvae  and late instar larvae while maintaining stable model dynamics.
\end{enumerate}
\subsection{Sensitivity Analysis of the Reproduction Number}
The analysis of our model demonstrates that the basic reproduction number acts as a key threshold indicator governing whether FAW infestations die out or persist during an outbreak. Because the parameters used in constructing the model are either estimated or adapted from existing literature, it is essential to examine how sensitive the basic reproduction number is to changes in these parameters. Such an examination provides insights into the uncertainties surrounding their values. To evaluate the contribution and significance of each parameter in shaping the basic reproduction number, we recall the normalized forward sensitivity index of $R_0$ given by \cite{alemenh2019ecoepidemiological}
\[
\Lambda_\rho^{R_0} = \frac{\partial R_0}{\partial \rho} \times \frac{\rho}{R_0},\tag{56}
\]
where $R_0$ represents the basic reproductive number and $\rho$ is the set of basic parameters whose sensitivities are as follows:
{\small
\begin{equation}\tag{57}\label{equation57}
\left\{
\begin{aligned}
\Lambda_{k}^{R_0}&=1, &
\Lambda_{a_1}^{R_0}&=\frac{a_1\gamma_1\mu_3}{a_1\gamma_1\mu_3+a_2\gamma_2\mu_2}\ge0, &
\Lambda_{a_2}^{R_0}&=\frac{a_2\gamma_2\mu_2}{a_1\gamma_1\mu_3+a_2\gamma_2\mu_2}\ge0,\\
\Lambda_{\gamma_1}^{R_0}&=\frac{a_1\gamma_1\mu_3}{a_1\gamma_1\mu_3+a_2\gamma_2\mu_2}\ge0, &
\Lambda_{\gamma_2}^{R_0}&=\frac{a_2\gamma_2\mu_2}{a_1\gamma_1\mu_3+a_2\gamma_2\mu_2}\ge0, &
\Lambda_{\mu_2}^{R_0}&=-\frac{a_1\gamma_1\mu_3}{a_1\gamma_1\mu_3+a_2\gamma_2\mu_2}\le0,\\
\Lambda_{\mu_3}^{R_0}&=-\frac{a_2\gamma_2\mu_2}{a_1\gamma_1\mu_3+a_2\gamma_2\mu_2}\le0.
\end{aligned}
\right.
\end{equation}
}
Using the baseline values in Table (\ref{table4}), the numerical outcomes for the mathematical expression in equations \eqref{equation57} were obtained and presented in Table (\ref{table3}) and the graphic illustration is in Figure (\ref{fig:2}).
\begin{table}[H]
\centering
\caption{\textbf{Normalized forward sensitivity indices of $R_0$ with respect to the model parameters.}}
\centering
\small
\renewcommand{\arraystretch}{0.6}
\begin{tabular}{llcccccccc}
\toprule
\multicolumn{2}{l}{\textbf{Parameter}} & $k$ & $a_1$ & $a_2$ & $\gamma_1$ & $\gamma_2$ & $\mu_2$ & $\mu_3$ \\
\midrule
\multicolumn{2}{l}{\textbf{Sensitivity index}} & +1 & +0.01 & +0.99 & +0.01 & +0.99 & -0.01 & -0.99 \\
\bottomrule
\end{tabular}
\label{table3}
\end{table}
The results indicate that the parameters $k$, $a_2$ and $\gamma_2$ have the highest positive sensitivity indices, implying that small increases in these parameters lead to significant increase in $R_0$. Among them, $k$ shows the highest sensitivity, implying that any proportional increase in maize leads to an equivalent proportional rise in the FAW population. Similarly, the parameters $a_2$ and $\gamma_2$ , associated with late instar larvae feeding and damage rates, significantly contribute to the increase in $R_0$. This emphasizes the critical role of late instar larvae in driving infestation intensity and crop damage. In contrast, parameters $\mu_2$ and $\mu_3$ represent the natural mortality rates of early and late instar larvae respectively, exhibiting negative sensitivity indices. This indicates that increasing larval mortality reduces $R_0$, thereby helping to control the FAW spread. However, their influence is comparatively smaller than that of $k$, $a_2$ and $\gamma_2$. Figure \ref{fig:2} presents the sensitivity analysis of the basic reproduction number with respect to key model parameters. 
\begin{figure}[h!]
    \centering
     \includegraphics[width=0.58\linewidth]{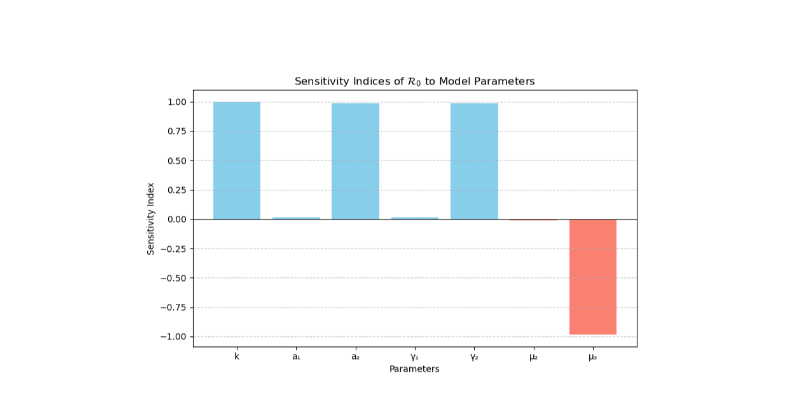}
    \caption{\textbf{Sensitivity analysis results showing the influence of model parameters on $R_0$.}}
    \label{fig:2}
\end{figure}
\begin{table}[h!]
\centering
\caption{\textbf{Parameter values for Systems \eqref{equation7} and \eqref{equation8}.}}
\resizebox{\textwidth}{!}{%
\begin{tabular}{|c|m{13.0cm}|c|c|}
\hline
\textbf{Parameter} & \textbf{Description} & \textbf{Baseline Value} & \textbf{Source} \\
\hline
$b$ & Intrinsic growth rate of maize plant population at $t=0$ & 0.05 $\mathrm{day}^{-1}$ & \cite{daudi2021mathematical} \\
\hline
$k$ & Maximum number of maize plants at $t=0$ & 500 plants & \cite{daudi_modelling_2021} \\
\hline
$\gamma_1$ & Maize destruction rate by early instar larvae in Period I & $9\times10^{-5}\,\mathrm{larva}^{-1}\,\mathrm{day}^{-1}$ & assumed \\
\hline
$\gamma_2$ & Maize destruction rate by late instar larvae in Period I & 0.00279 $\mathrm{larva}^{-1}\,\mathrm{day}^{-1}$ & assumed \\
\hline
$\gamma_3$ & Maize destruction rate by early instar larvae in Period II & 0.005 $\mathrm{larva}^{-1}\,\mathrm{day}^{-1}$ & assumed \\
\hline
$\gamma_4$ & Maize destruction rate by late instar larvae in Period II & 0.0167 $\mathrm{larva}^{-1}\,\mathrm{day}^{-1}$ & assumed \\
\hline
$\alpha$ & Proportion of female adult moths & 0.5 & \cite{daudi2021mathematical} \\
\hline
$\lambda$ & Egg laying rate & 0.0417 eggs\, $\mathrm{adult}^{-1}\,\mathrm{day}^{-1}$ & \cite{daudi_modelling_2021} \\
\hline
$\beta$ & Hatching rate of FAW eggs into early instar & 0.071 $\mathrm{day}^{-1}$ & \cite{daudi_modelling_2021} \\
\hline
$a_1$ & Conversion rate of maize into feed by early instar larvae  (Period I) & 0.001 larva\,$\mathrm{plant}^{-1}$ & assumed \\
\hline
$a_2$ & Conversion rate of maize into feed by late instar larvae (Period I) & 0.0101  larva\,$\mathrm{plant}^{-1}$ & assumed \\
\hline
$a_3$ & Conversion rate of maize into feed by early instar larvae (Period II) & 0.004  larva\,$\mathrm{plant}^{-1}$ & assumed \\
\hline
$a_4$ & Conversion rate of maize into feed by late instar larvae (Period II) & 0.008  larva\,$\mathrm{plant}^{-1}$ & assumed \\
\hline
$\delta_1$ & Transition rate of early to late instar larvae & 0.071 $\mathrm{day}^{-1}$ & \cite{daudi_modelling_2021} \\
\hline
$P_1$ & Rate of cannibalism between early and late instar larvae & 0.03 larva\,$\mathrm{day}^{-1}$ & assumed \\
\hline
$P_2$ & Rate of cannibalism within late instar larvae & 0.0099 larva\,$\mathrm{day}^{-1}$  & assumed \\
\hline
$\mu_1$ & Mortality rate of egg & 0.0009 $\mathrm{day}^{-1}$ & assumed \\
\hline
$\mu_2$ & Natural mortality rate of early instar larvae  & 0.0071 $\mathrm{day}^{-1}$ & \cite{daudi_modelling_2021} \\
\hline
$\mu_3$ & Natural mortality rate of late instar larvae  & 0.0071 $\mathrm{day}^{-1}$ & \cite{daudi_modelling_2021} \\
\hline
$\mu_4$ & Natural mortality rate of adult moth  & 0.1 $\mathrm{day}^{-1}$ & assumed \\
\hline
$\delta_2$ & Rate at which late instar larvae develop into adult moth & 0.071 $\mathrm{day}^{-1}$ & \cite{daudi_modelling_2021} \\
\hline
\end{tabular}}
\label{table4}
\end{table}
\newpage
The simulation results presented in Figures \ref{fig:3} and \ref{fig:4} illustrate the temporal dynamics of adult moths, eggs and larval instars, as well as their impact on maize growth. To perform the simulations for systems \eqref{equation7} and \eqref{equation8}, the initial conditions for period I were set as: $M_1(0) = 500,\; E(0) = 0,\; L_1(0) = 0,\; L_2(0) = 0,\; A(0) = 40$. The initial conditions for period II at $t_1=63$ were obtained from the terminal values of period I, ensuring continuity between the two stages. From the numerical simulation, these values are: $M_2(63) = 394 ,\; E(63) = 73,\; L_1(63) = 2.4,\; L_2(63) = 4.4,\; A(63) = 3$. System \eqref{equation7} was simulated over the interval  $t=0$ to $t_1=63$, while system \eqref{equation8} was simulated over the interval $t_1=63$ to $t_2= 160$. This approach enables a consistent evaluation of the interaction between maize and FAW population across both vegetative and reproductive stages.
\begin{figure}[h!]
    \centering
    \begin{subfigure}[t]{0.40\linewidth}
        \centering
       \includegraphics[width=\linewidth]{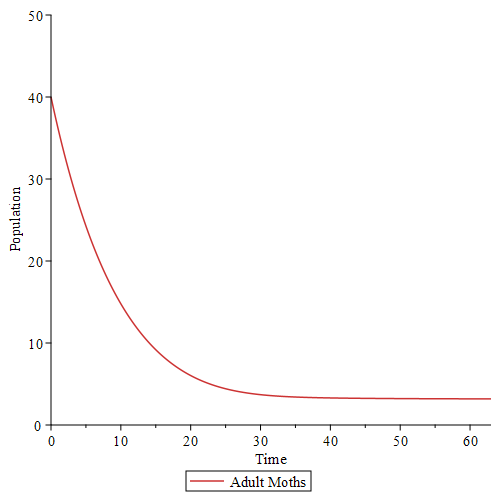}
        \caption{}
        \label{fig:adult}
    \end{subfigure}
    \hfill
    \begin{subfigure}[t]{0.40\linewidth}
        \centering
         \includegraphics[width=\linewidth]{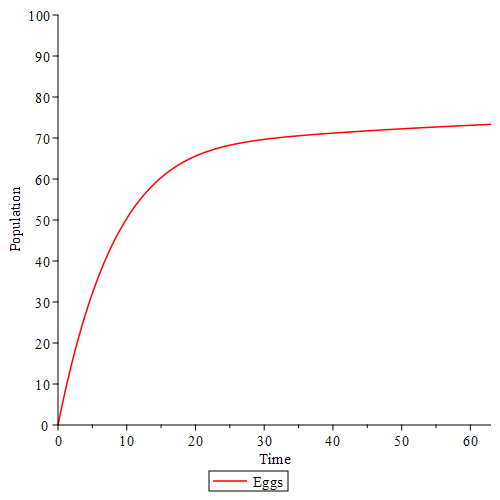}
        \caption{}
        \label{fig:egg}
    \end{subfigure}    

    \begin{subfigure}[t]{0.40\linewidth}
        \centering
        \includegraphics[width=\linewidth]{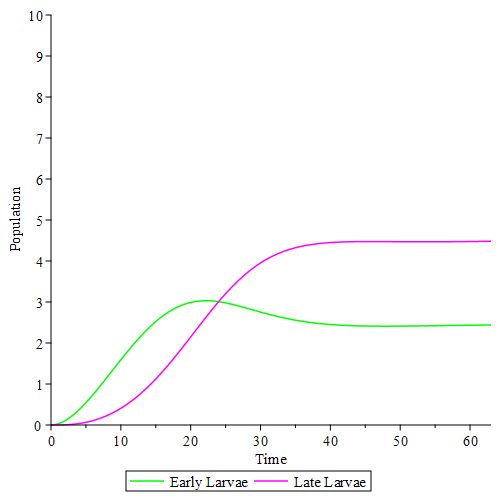}
        \caption{}
        \label{fig:larvae}
    \end{subfigure}
    \hfill
    \begin{subfigure}[t]{0.40\linewidth}
        \centering
        \includegraphics[width=\linewidth]{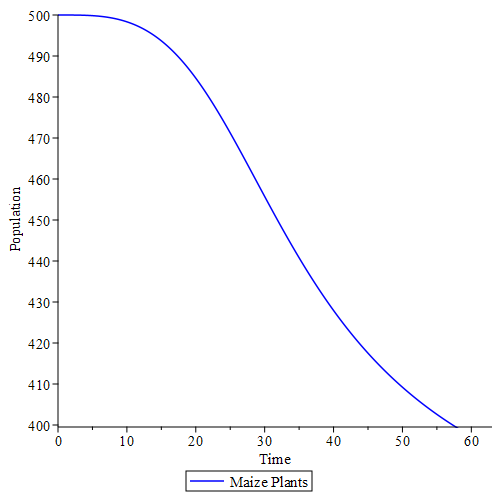}
        \caption{}
        \label{fig:maize}
    \end{subfigure}
    \caption{\textbf{ Numerical simulations of system \eqref{equation7} illustrating the interactions between the maize plant population and FAW population during the vegetative stage. Using the initial conditions and baseline parameter values in Table \ref{table4}, the computed basic reproduction number is $R_0=1.9908$. Since $R_0>1$, the FAW population persists, leading to a progressive decline in maize plant population.}}
    \label{fig:3}
\end{figure}
\\
Figure \eqref{fig:3} illustrates the dynamics of the maize plant population and the FAW population during the vegetative stage over the time interval $t=0$ to $t_1=63$. In panel (a), the adult moth population declines rapidly from 40 adult moth to 4.4 within the first 30 days, due to natural mortality. Panel (b) shows a sharp increase in egg population, reaching 73 eggs by day 20 before stabilizing, reflecting sustained oviposition by the surviving adult moths. In panel (c), the early instar larvae peak at 3 larvae by day 20, then decline slightly to about 2.4 larvae as they mature. The late instar larvae increase more gradually, surpassing the early instars around day 25 and stabilizing at approximately 4.4 larvae, indicating accumulation of more destructive stages. Panel (d) shows the maize plant population declining from 500 to 394 by day 58 representing a loss of about 21\% of the initial maize population. Biologically, this highlights that the late instar larvae are primarily responsible for severe crop damage due to their higher feeding rates. Overall, the results indicate that maize loss is primarily driven by the buildup of late instar larvae, highlighting the need of early intervention.
\begin{figure}[hbtp!]
    \centering
    \begin{subfigure}[t]{0.40\linewidth}
        \centering
      \includegraphics[width=\linewidth]{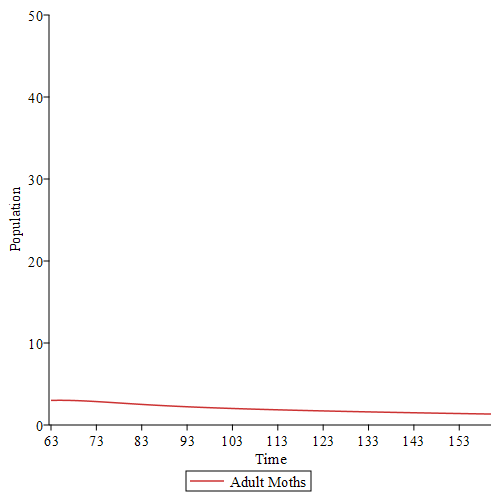}
        \caption{}
        \label{fig:adult}
    \end{subfigure}
    \hfill
    \begin{subfigure}[t]{0.40\linewidth}
        \centering
            \includegraphics[width=\linewidth]{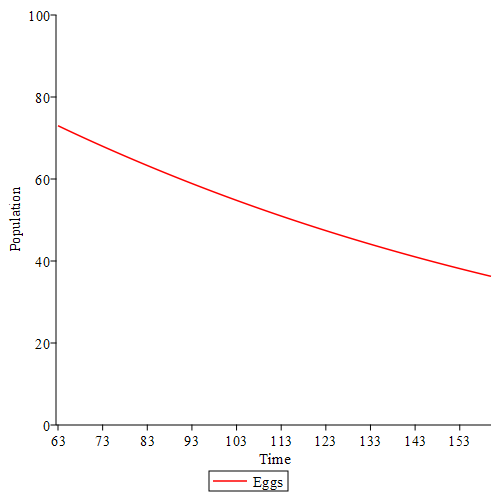}
        \caption{}
        \label{fig:egg}
    \end{subfigure}
    
    \begin{subfigure}[t]{0.40\linewidth}
        \centering
          \includegraphics[width=\linewidth]{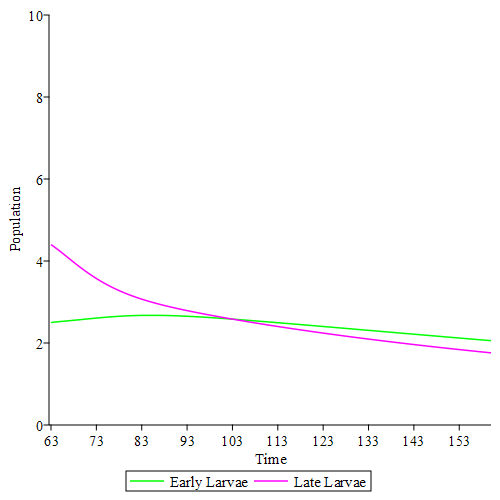}
        \caption{}
        \label{fig:larvae}
    \end{subfigure}
    \hfill
    \begin{subfigure}[t]{0.40\linewidth}
        \centering
             \includegraphics[width=\linewidth]{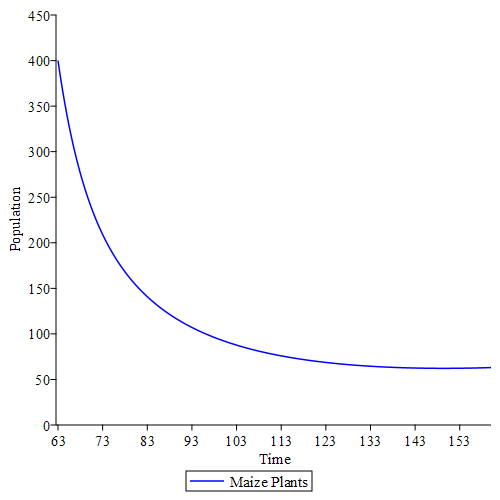}
        \caption{}
        \label{fig:maize}
    \end{subfigure}
    \caption{\textbf{  Numerical simulations of system \eqref{equation8} illustrating the interactions between the maize plant population and FAW population during the reproductive stage.  The computed basic reproduction number, $R_0=10.8169$, indicates sustained FAW infestation,  resulting in persistent larval instars and a marked decline in maize plant population.}}
    \label{fig:4}
\end{figure}
\\
Figure \ref{fig:4} presents the dynamics of the maize and FAW population during the reproductive stage  over the time interval $t_1=63$ to $t_2=160$. In panel (a), the adult moth population remains low, fluctuating between approximately 2 and 3 adult moths before gradually declining, indicating limited recruitment and continued natural mortality. Panel (b) shows a steady decline in egg population from approximately 73 eggs at $t=63$ to about 36 eggs by $t=160$, reflecting reduced oviposition and unfavorable conditions for reproduction. In panel (c), the early instar larvae initially increase slightly from 2.4 larvae before gradually declining to approximately 2 larvae. The late instar larvae, however, show a continuous decline 4.4 to 1.7 larvae at $t=160$. This behavior suggests reduced larval recruitment together with continued progression through the developmental stages. Panel (d) shows a sharp decline in maize plant population from approximately 394 at $t=63$ to 63 by $t=160$, representing a loss of approximately 85\% of the initial maize plant population. Overall, although the FAW population declines, the persistence of the late instar larvae continues to cause substantial maize loss throughout the reproductive stage.
\\
Figures \ref{fig:5} and \ref{fig:6} represent the behavior of the maize plant population and FAW population after applying controls. As noted by Daudi et al. \cite{daudi_modelling_2021}, pesticides act in a density-dependent manner, reducing FAW population through chemical toxicity but have no direct effect on maize plants. To reduce the damage caused by FAW on maize production, we developed and analyzed an optimal control framework targeting the egg, early instar larvae and late instar larvae. The objective of this approach is to minimize the pest population  and the associated maize damage.
\\
The optimality systems for both the vegetative and reproductive stages were solved numerically using the forward-backward sweep method. The state equations were integrated forward in time, while the adjoint equations were integrated backward in time using the forward Euler scheme with a uniform time step of $\Delta t=0.1$ days. For the reproductive stage, the adjoint equations were integrated backward from the terminal time $t=160$ to the initial time $t=63$, starting from the prescribed terminal transversality conditions. The weighting coefficients in the objective function were chosen as $C_1=2$, $C_2=4$ and $C_3=7$, assigning progressively greater importance to reducing eggs, early instar larvae and late instar larvae, respectively, to reflect their increasing contribution to maize damage. The control cost weights were selected as $W_1=300$ and $W_2=450$. A convergence tolerance of $10^{-4}$ was adopted and a maximum of $100$ forward-backward sweep iterations were permitted. The iterative procedure terminated when the maximum absolute difference between successive control profiles was less than $10^{-4}$, otherwise, the algorithm stopped after reaching the prescribed maximum number of iterations. Throughout the optimization process, the controls were projected onto their admissible bounds, $0\leq u_1(t)\leq u_{1,\max}$ and $0\leq u_2(t)\leq u_{2,\max}$. System \eqref{equation46} was simulated with control bounds \( 0 < u_1(t) \leq 0.4 \) and \( 0 < u_2(t) \leq 0.5 \) over the interval $t\in[0,63]$ while system \eqref{equation47} was simulated with control bounds \( 0 < u_1(t) \leq 0.1 \) and \( 0 < u_2(t) \leq 0.2 \)  over the interval $t\in [63,160]$. Thus, the vegetative and reproductive stages span 63 and 97 days, respectively. The resulting optimal control strategies effectively suppressed the FAW population while maintaining the maize plant population close to its carrying capacity, as illustrated in Figures \ref{fig:5} and \ref{fig:6}.
\begin{figure}[!ht]
\centering
\begin{subfigure}[t]{0.40\textwidth}
    \centering
    \includegraphics[width=\linewidth]{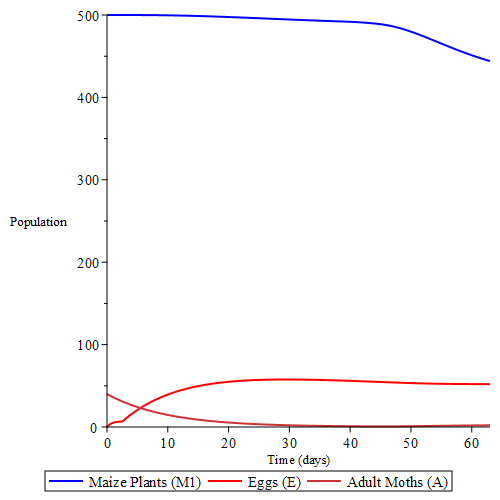}
    \caption{}
    \label{fig:rep_faw}
\end{subfigure}
\hfill
\begin{subfigure}[t]{0.40\textwidth}
    \centering
   \includegraphics[width=\linewidth]{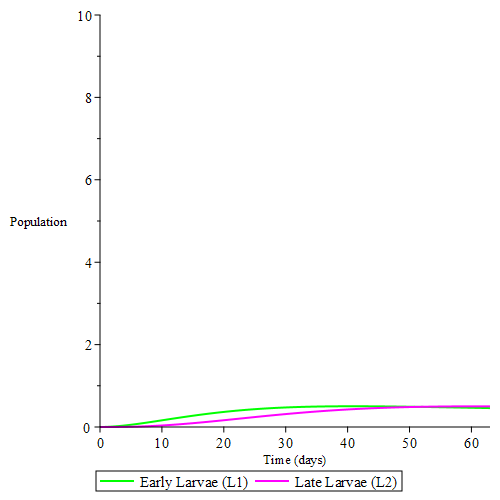}
    \caption{}
    \label{fig:rep1_faw}
\end{subfigure}\label{fig:5} 

\caption{\textbf{Impact of optimal control measures on the maize plant population and FAW population during the vegetative stage.}}
\label{fig:5}
\end{figure}
\\
Figure \ref{fig:5} illustrates the dynamics of the maize plant population and FAW population during the vegetative stage under the optimal control strategy. As shown in Figure \ref{fig:5}(a), the maize plant population remained close to its carrying capacity, with approximately 450 plants surviving at the end of the simulation, representing a substantial improvement over the uncontrolled scenario. The egg population increased slightly during the early stages due to continued oviposition by adult moths already present before their population was effectively suppressed. As shown in Figure \ref{fig:5}(b), both the early instar larvae and late instar larvae remained low throughout the simulation. Although a small number of eggs continued to hatch, the combined effects of traditional  and chemical controls prevented larval population from increasing substantially. Consequently, feeding damage was minimized, allowing more maize plants to survive during the vegetative stage. The optimal control profiles were initially maintained at their maximum admissible levels to suppress the developing infestation. The traditional control $u_1(t)$ remained at 0.4 for about 40 days, while the chemical control $u_2(t)$ remained at 0.5 for about 45 days. Thereafter, both controls gradually declined as the pest population was progressively brought under control, approaching zero by the end of the vegetative stage.
\begin{figure}[!ht]
\centering
\begin{subfigure}[t]{0.40\textwidth}
    \centering
    \includegraphics[width=\linewidth]{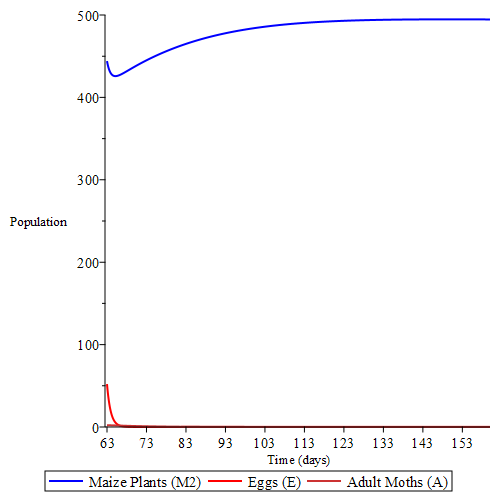}
    \caption{}
    \label{fig:rep_faw}
\end{subfigure}
\hfill
\begin{subfigure}[t]{0.40\textwidth}
    \centering
       \includegraphics[width=\linewidth]{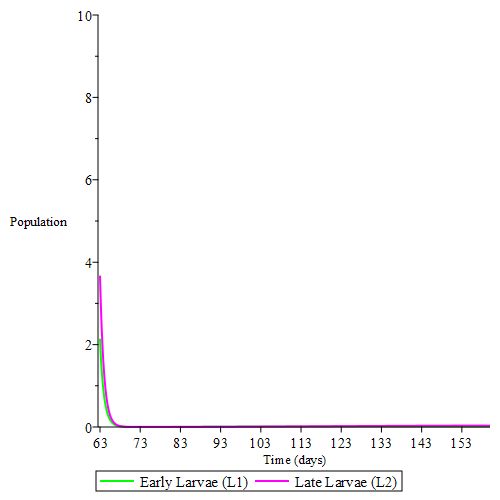}
    \caption{}
    \label{fig:rep1_faw}
\end{subfigure}\label{fig:6}

\caption{\textbf{Impact of optimal control measures on the maize plant population  and FAW population  during the reproductive stage.}}
\label{fig:6}
\end{figure}
\\
Figure \ref{fig:6} presents the dynamics of the maize plant population and FAW population during the reproductive stage under the optimal control strategy. As shown in Figure \ref{fig:6}(a), the egg and adult moth populations declined rapidly to negligible levels following the continuation of the control measures from the vegetative stage. This rapid decline reflects the cumulative effect of sustained handpicking, chemical control, natural mortality and the reduced recruitment of new adult moths resulting from the suppression of the larval instars. Consequently, oviposition was substantially reduced, preventing further population establishment. Figure \ref{fig:6}(b) shows that both the early instar larvae  and late instar larvae were effectively eliminated within a short period, indicating that the integrated control strategy successfully interrupted the FAW life cycle. With feeding pressure largely removed, maize damage was minimized throughout the reproductive stage, resulting in a substantially healthier crop compared with the uncontrolled scenario. The optimal control profiles were initially maintained at their maximum admissible levels $u_1=0.1$ and $u_2=0.2$ to eliminate the remaining pest population from the vegetative stage. The traditional control remained at its maximum for approximately 5 days, while the chemical control was maintained for about 10 days before both control efforts gradually declined as the pest population was completely suppressed. By the end of the reproductive stage, the control efforts had approached zero, reflecting the absence of further infestation.
\newpage

To assess model robustness across the vegetative and reproductive stages, a one-at-a-time (OAT) parameter variation analysis was performed by varying the maize feeding parameters $a_1$, $a_2$, $a_3$ and $a_4$, maize destruction rates $\gamma_1$, $\gamma_2$, $\gamma_3$ and $\gamma_4$ and larval mortality rates $\mu_2$ and $\mu_3$, individually by $\pm20\%$ from their baseline values, while keeping all other parameters fixed. The resulting population dynamics were compared with the corresponding baseline simulations for both stages. Overall, the OAT analysis showed that the model remained qualitatively robust to moderate parameter perturbations, with the maize population exhibiting the most noticeable quantitative changes. The corresponding results for the vegetative and reproductive stages are presented in Tables \eqref{Table 5} and \eqref{Table 6}, respectively.
\begin{table}[htbp]
\centering
\caption{\textbf{Vegetative stage output changes under ±20\% OAT parameter variation.}}
\label{Table 5}
\begin{tabular}{lcccccc}
\hline
\textbf{Parameter} & \textbf{Change} & \textbf{Maize} & \textbf{Eggs} & \textbf{Early instar larvae} & \textbf{Late instar larvae} & \textbf{Adult} \\
\hline
 & \textbf{Baseline (0\%)} & 394.5105 & 73.3391 & 2.4402 & 4.4814 & 3.1881 \\
\hline
$a_1$ & $+20\%$ & 394.5071 & 73.3398 & 2.4403 & 4.4815 & 3.1882 \\
\hline
$a_1$ & $-20\%$ & 394.5138 & 73.3383 & 2.4402 & 4.4812 & 3.1880 \\
\hline
$a_2$ & $+20\%$ &  392.9237 & 73.6393 & 2.4236 & 4.5563 & 3.2396 \\
\hline
$a_2$ & $-20\%$ & 396.0737 & 73.0442 & 2.4570 & 4.4075 & 3.1375 \\
\hline
$\gamma_1$ & $+20\%$ & 394.1299 & 73.3388 & 2.4404 & 4.4812 & 3.1880 \\
\hline
$\gamma_1$ & $-20\%$ & 394.8911 & 73.3393 & 2.4401 & 4.4816 & 3.1882 \\
\hline
$\gamma_2$ & $+20\%$ & 373.6952 & 73.5979 & 2.4287 & 4.5366 & 3.2280 \\
\hline
$\gamma_2$ & $-20\%$ & 415.4209 & 73.0707 & 2.4536 & 4.4203 & 3.1450 \\
\hline
$\mu_2$ & $+20\%$ & 395.4355 & 73.1716 & 2.4510 & 4.4353 & 3.1574 \\
\hline
$\mu_2$ & $-20\%$ & 393.5747 & 73.5085 & 2.4295 & 4.5280 & 3.2192 \\
\hline
$\mu_3$ & $+20\%$ & 403.2960 & 71.7527 & 2.5472 & 4.0475 & 2.8985 \\
\hline
$\mu_3$ & $-20\%$ & 384.6634 & 75.1280 & 2.3328 & 4.9767 & 3.5176 \\
\hline
\end{tabular}
\end{table}

\begin{table}[htbp]
\centering
\caption{\textbf{Reproductive stage output changes under ±20\% OAT parameter variation.}}
\label{Table 6}
\begin{tabular}{lcccccc}
\hline
\textbf{Parameter} & \textbf{Change} & \textbf{Maize} & \textbf{Eggs} & \textbf{Early instar larvae} & \textbf{Late instar larvae} & \textbf{Adult} \\
\hline
 & \textbf{Baseline (0\%)} & 63.0360 & 36.2653 & 2.0569 & 1.7599 &  1.3357 \\
\hline
$a_3$ & $+20\%$ &62.6583 & 36.2699 & 2.0597 &1.7624 & 1.3376 \\
\hline
$a_3$ & $-20\%$ & 63.4169 & 36.2607 & 2.0540 & 1.7573 & 1.3337 \\
\hline
$a_4$ & $+20\%$ &  61.0771 & 36.3029 & 2.0505 & 1.7756 & 1.3475 \\
\hline
$a_4$ & $-20\%$ & 65.9066 & 36.2815 & 2.7169 & 1.6808 & 1.0163 \\
\hline
$\gamma_3$ & $+20\%$ & 54.2716 & 36.2578 & 2.0629 & 1.7481 & 1.3276 \\
\hline
$\gamma_3$ & $-20\%$ & 72.6887 & 36.2732 & 2.0497 & 1.7725 & 1.3442 \\
\hline
$\gamma_4$ & $+20\%$ & 37.2212 & 36.2527 & 2.0634 & 1.7282 & 1.3136 \\
\hline
$\gamma_4$ & $-20\%$ & 99.4202 & 36.2724 &2.0533 & 1.7919 & 1.3577 \\
\hline
$\mu_2$ & $+20\%$ & 64.4437 & 36.2504 & 2.0414 & 1.7445 & 1.3244 \\
\hline
$\mu_2$ & $-20\%$ & 61.6386 & 36.2804 & 2.0724 & 1.7754 & 1.3470 \\
\hline
$\mu_3$ & $+20\%$ & 73.3429 & 36.1011 & 2.1203 & 1.6274 &1.2363 \\
\hline
$\mu_3$ & $-20\%$ & 52.4815 & 36.4508 & 1.9892 & 1.9128 & 1.4500 \\
\hline
\end{tabular}
\end{table}
\newpage

\section{Discussion and Concluding Remarks }\label{sect5}
In this paper, we developed a stage-structured mathematical model to investigate the effect of FAW infestation on maize farming, during the vegetative and reproductive stages. Analysis revealed that each model admits a unique solution that remains positive and bounded for all \(t\geq 0 \). Furthermore, the model was shown to possess four equilibrium points: the trivial equilibrium, the non-trivial equilibrium, the maize extinction equilibrium  and the coexistence equilibrium. The local stability was established using the Jacobian matrix and the Routh-Hurwitz criterion under the conditions stated in Theorems \ref{theorem3}-\ref{theorem6}, while the global stability was investigated using Lyapunov functions following the approach presented in subsubsection \ref{section2.1.6}.
Sensitivity analysis identified the parameters with the greatest influence on the FAW population dynamics.
Numerical simulations further indicate substantial reductions in the maize plant population during the vegetative and reproductive stages of maize production due to increased egg production and larval population. The severe damage to the maize plant population caused by high egg, larval and adult moth populations  motivated the extension of the models to incorporate optimal control strategies, including traditional methods like handpicking and chemical pesticides. Simulation results demonstrated that the proposed control strategy effectively suppressed the FAW population, resulting in an increase in the maize plant population. Compared to existing studies, the proposed model provides a more realistic representation of the FAW population dynamics by explicitly separating the larval population into early and late instar larvae. This distinction captures differences in feeding behavior and developmental progression that are often neglected in previous models. Nevertheless, the proposed model does not account for environmental factors such as temperature and rainfall, which are known to influence the development, survival and reproduction of the FAW population. Although chemical pesticides remain an effective means of controlling FAW, their excessive use may pose risks to human health and the environment. Therefore, we recommend integrating environmentally friendly approaches, such as mating disruption techniques, into future FAW management strategies. By reducing successful mating and subsequent egg production, this approach has the potential to lower larval population while minimizing reliance on chemical pesticides. Future studies should also incorporate Cost-effective analysis, using measures such as the Incremental Cost-Effective Ratio (ICER), to assess the economic efficiency of the proposed control strategies.
Collectively, these extensions would contribute to more sustainable FAW management and improve maize productivity.
\subsection*{Acknowledgment}
The authors sincerely acknowledge the encouragement and support received from members of the department. They are also grateful for the valuable review comments and constructive suggestions provided. The authors deeply appreciate this generous support and insightful feedback.

\subsection*{Source of funding}
This research did not receive any external funding and all costs associated with the study were fully covered by the authors.

\subsection*{Data Availability Statement}
The data used to support the model analysis were obtained from existing literature. Some parameter values were adopted from published studies, while others were reasonably assumed. 
\subsection*{Conflict of interest}
The authors declare that there are no conflicts of interest associated with the publication of this paper.
\newpage
\renewcommand{\refname}{References}
\bibliographystyle{unsrt}
\bibliography{ref}

\newpage
\end{document}